\documentclass[11pt]{article}
\usepackage{jheppub}
\usepackage{bm,bbm}
\usepackage{booktabs,amsmath}
\usepackage{accents}
\usepackage{multirow}
\usepackage{graphics}
\usepackage{braket}
\usepackage{mathtools}
\usepackage{comment}
\usepackage{autobreak}
\usepackage{subcaption}
\usepackage{enumerate}
\usepackage{blkarray}
\usepackage{longtable}
\usepackage{enumitem}
\usepackage{soul}

\usepackage{tikz}

\usetikzlibrary{patterns}
\usetikzlibrary{arrows,shapes,positioning}
\usetikzlibrary{decorations.markings}
\usetikzlibrary{decorations.pathmorphing}
\tikzset{snake it/.style={decorate, decoration=snake}}

\hypersetup{
    colorlinks=true,
    linkcolor=blue,
    filecolor=magenta,      
    urlcolor=cyan,
    pdfpagemode=FullScreen,
}

\usepackage{arydshln}
\usepackage{mathtools}

\edef\restoreparindent{\parindent=\the\parindent\relax}
\usepackage{parskip}
\restoreparindent
\usepackage{ascmac}
\usepackage{mathrsfs}

\usepackage{amsthm}
\newtheoremstyle{break}
  {\topsep}{\topsep}%
  {\upshape}{}%
  {\bfseries}{}%
  {\newline}{}%
\theoremstyle{break}
\newtheorem{proposition}{Proposition}[section]
\newtheorem{theorem}[proposition]{Theorem}
\newtheorem{corollary}[proposition]{Corollary}

\def\tr{{\rm tr}}
\def\Tr{{\rm Tr}}
\def\d{{\rm d}}
\def\i{{\rm i}}

\def\CC{{\cal C}}

\def\CE{{\cal E}}

\def\CH{{\cal H}}

\def\BC{\mathbb{C}}
\def\BF{\mathbb{F}}
\def\BN{\mathbb{N}}

\def\BQ{\mathbb{Q}}
\def\BR{\mathbb{R}}

\def\BZ{\mathbb{Z}}

\def\Ba{\mathbf{a}}
\def\Bb{\mathbf{b}}
\def\Bc{\mathbf{c}}
\def\Bu{\mathbf{u}}
\def\Bv{\mathbf{v}}

\def\d{\mathrm{d}}

\def\SU{\mathrm{SU}}

\def\U{\mathrm{U}}

\title{
Narain CFTs from nonbinary stabilizer codes}

\author[a,b]{Yasin Ferdous Alam,} 
\author[a,c]{Kohki Kawabata,}
\author[a]{Tatsuma Nishioka,}
\author[d]{Takuya Okuda}
\author[c]{and Shinichiro Yahagi}

\affiliation[a]{Department of Physics, Osaka University,\\
Machikaneyama-Cho 1-1, Toyonaka 560-0043, Japan}

\affiliation[b]{Physics Department, University of Texas, Austin,\\
Austin TX 78712, USA}

\affiliation[c]{Department of Physics, Faculty of Science,
The University of Tokyo,\\
Bunkyo-Ku, Tokyo 113-0033, Japan}

\affiliation[d]{Graduate School of Arts and Sciences, The University of Tokyo, Komaba,\\
Meguro-ku, Tokyo 153-8902, Japan}

\preprint{OU-HET-1193, UT-Komaba/23-8}

\abstract{
We generalize the construction of Narain conformal field theories (CFTs) from qudit stabilizer codes to the construction from quantum stabilizer codes over the finite field of prime power order ($\mathbb{F}_{p^m}$ with $p$ prime and $m\geq 1$) or over the ring $\mathbb{Z}_k$ with $k>1$.
Our construction results in rational CFTs, which cover a larger set of points in the moduli space of Narain CFTs than the previous one.
We also propose a correspondence between a quantum stabilizer code with non-zero logical qubits
and a finite set of Narain CFTs.
We illustrate the correspondence with well-known stabilizer codes.
}

\begin{document} 
\maketitle
\flushbottom

\section{Introduction}
Narain conformal field theories (CFTs) are a class of two-dimensional CFTs describing the geometry of the spacetime in string theory \cite{Narain:1985jj,Narain:1986am}.
They are characterized by a set of vertex operators whose left- and right-moving momenta span so-called Narain lattices. 
Narain CFTs are bosonic, non-chiral, and modular invariant when the lattices are even, Lorentzian, and self-dual, respectively.
While they are simple theories of free compact bosons specified by the lattices, they have a rather large continuous moduli space and exhibit a rich structure such as symmetry enhancement and dualities \cite{Polchinski:1998rq}.

Narain CFTs have received renewed interest in recent years due to the revelation of their relation to quantum error-correcting codes (QECs) \cite{Dymarsky:2020qom}.
This relation is based on the novel construction of even self-dual Lorentzian lattices from qubit stabilizer codes \cite{Gottesman:1996rt}, which extends Construction A of even self-dual Euclidean lattices from classical error-correcting codes \cite{conway2013sphere}.
The CFTs whose Narain lattices are the even self-dual Lorentzian lattices built from qubit stabilizer codes are named Narain code CFTs.
Recently, there has been significant progress in the field of Narain code CFTs.
Generalizations have been made from qubit (binary) to qudit (non-binary) stabilizer codes \cite{Kawabata:2022jxt}, and a family of code CFTs have been constructed from quantum Calderbank-Shor-Steane (CSS) codes~\cite{Calderbank:1995dw,Steane:1996va}.
There are also various applications of Narain code CFTs, including the modular bootstrap \cite{Dymarsky:2020bps,Henriksson:2022dnu,Dymarsky:2022kwb}, the exploration of CFTs with large spectral gaps \cite{Furuta:2022ykh,Angelinos:2022umf}, and their relevance in the context of holographic duality \cite{Dymarsky:2020pzc} based on ensemble average of Narain CFTs \cite{Maloney:2020nni,Afkhami-Jeddi:2020ezh,aharony2023holographic}.
See also \cite{Gaiotto:2018ypj,Buican:2021uyp,Henriksson:2022dml,Yahagi:2022idq,Kawabata:2023nlt} for relevant works.

In this paper, we further extend the constructions of Narain code CFTs from quantum stabilizer codes over $\BF_p$ with $p$ prime \cite{Dymarsky:2020qom,Kawabata:2022jxt} to those from quantum stabilizer codes over finite field $\BF_{p^m}$ with $p$ prime and $m\ge 1$ and over the ring $\BZ_k$ with $k>1$.
Following the approach adopted in previous studies, we proceed to establish Narain code CFTs by constructing them from quantum codes in two steps:
\begin{enumerate}
    \item Map quantum stabilizer codes to their associated classical codes.
    \item Lift the resulting classical codes to even self-dual Lorentzian lattices to Narain code CFTs.
\end{enumerate}
The first step follows immediately from the result of \cite{ashikhmin2001nonbinary} for finite field $\BF_{p^m}$ and those of \cite{nadella2012stabilizer,guenda2014quantum} for finite (Frobenius) rings.
In the second step, we can leverage the construction of Narain code CFT from classical codes over $\BF_{p^m}$ in \cite{Yahagi:2022idq} for the finite field case.
To our best knowledge, however, a similar construction appears to be missing in the literature. 
Focusing on quantum stabilizer codes over $\BZ_k$, we use Construction A to lift the associated classical codes to lattices and determine the condition when the resulting lattices become even self-dual (Theorem \ref{theorem:self_duality_ring_lattice}).
Moreover, we show that certain CSS codes always satisfy the condition (Proposition \ref{prop:CSS_ring}) and yield a broader class of Narain code CFTs than those in \cite{Dymarsky:2020qom,Kawabata:2022jxt}.
For example, focusing on CSS codes associated with classical self-dual codes, Narain code CFTs constructed from the codes over $\BF_p$ with prime $p$ always have even central charges irrespective of the code length.
On the other hand, Narain code CFTs constructed from the CSS codes associated with classical self-dual codes over $\BZ_k$ with $k$ odd and over $\BF_{p^m}$ with $m$ odd can have odd central charges when the code length is odd.\footnote{Quantum stabilizer codes of non self-dual CSS type can yield Narain code CFTs with odd central charges \cite{Dymarsky:2020pzc,Kawabata:2022jxt}.}
Thus, our construction covers a larger (discrete) set of points in the moduli space of Narain CFTs than before and paves the way toward classifying the ensembles of Narain CFTs whose properties can be understood from those of the associated quantum codes.\footnote{
The larger moduli space covered by our Narain CFTs allows for a wider class of ensembles to average through a finer grain ($\BF_{p^m}$ and $\BZ_k$ as opposed to just $\BF_p$). 
In the limit of large order $q$ for Narain code CFTs constructed over a ring $R_q$, the ensemble approaches that of all Narain CFTs \cite{Angelinos:2022umf,Kawabata:2022jxt}. Averaging over an ensemble of Narain CFTs is holographically dual to Chern-Simons theories and taking the large order limit is expected to be equivalent to three-dimensional $\U(1)$-gravity \cite{Maloney:2020nni,Afkhami-Jeddi:2020ezh}. 
}

With the general construction in hand, we examine quantum stabilizer codes of short length and explore Narain code CFTs with small central charges.
We show that Narain code CFTs from stabilizer codes of length one over $\BZ_{k}$ give rise to a single compact boson theory whose radius square can be any rational number, $R^2 = 2m_1/m_2$ by choosing $k=m_1\,m_2$ where $m_1,m_2$ are coprime positive integers.
Since this construction is believed to exhaust all rational compact boson theories \cite{Ginsparg:1987eb,DiFrancesco:1997nk,Gukov:2002nw} by varying $m_1$ and $m_2$, we succeed in realizing
 all rational points in the moduli space of Narain CFTs with central charge one by our construction.\footnote{Any $c=1$ Narain CFT can be obtained from a code CFT by acting an appropriate $\text{O}(1,1, \BR)$ transformation \cite{Angelinos:2022umf}, while our construction provides all $c=1$ rational Narain CFTs without relying on such a transformation.}
We also consider a few examples of stabilizer codes over $\BZ_k$ and over $\BF_{p^m}$ of length two and construct their Narain code CFTs with central charge two.
We determine representatives of the metric and anti-symmetric tensor describing the geometry of the target space in Narain CFTs, and show explicitly that the Narain code CFTs with central charge two we consider as examples are rational based on the criteria given in \cite{Gukov:2002nw}.
Note that there are Narain code CFTs that are not rational as exemplified by \cite{Dymarsky:2021xfc}.
We show that the Narain code CFTs we construct in this paper are rational as their partition functions can be written as a finite sum of a product of characters.
Incidentally, a recent paper \cite{Furuta:2023xwl} obtains sufficient conditions on the metric and B-field for Narain CFTs to be rational and have a code construction, which is also derived from the character decomposition of the torus partition functions. Our results are consistent with those of \cite{Furuta:2023xwl}.

Another point of interest in this paper is the extension of the one-to-one correspondence between (binary) quantum stabilizer codes and Narain code CFTs to a one-to-many correspondence.
In the former correspondence, the relevant stabilizer code always has zero logical qudits.
Practically useful stabilizer codes, and CSS codes in particular, have non-zero logical qubits.
We show that a CSS code with a non-zero number of qubits can be naturally associated with a set of several Narain code CFTs.
We illustrate the one-to-many correspondence with two famous CSS codes, the Shor code~\cite{PhysRevA.52.R2493} and the Steane code \cite{steane1996multiple}.

This paper is organized as follows.
In section \ref{ss:review_stabilizer}, we review the stabilizer formalism of quantum codes over finite Frobenius rings with emphasis on the profound relation bridging between quantum stabilizer codes and classical codes.
We also introduce CSS codes which provide systematic constructions of quantum stabilizer codes from a pair of classical codes satisfying certain conditions.
In section \ref{ss:review_narain}, 
we review the description of a Narain CFT in terms of the momentum lattice and also review the rationality conditions for theories with small central charges.
In section \ref{ss:ringcodes}, we focus on CSS codes over $\BZ_k$ and construct their Narain code CFTs through Construction A of the Lorentzian lattices.
We also examine several Narain code CFTs with central charges one and two, and study their rationality.
In section \ref{ss:finite_field}, we turn to CSS codes over $\BF_{p^m}$ and the associated Narain code CFTs.
A few examples are worked out and shown to reduce to known CFTs.
In section \ref{sec:non-zero-qubits}, we switch gears and propose an extended correspondence between a quantum stabilizer code with non-zero logical qubits and a set of Narain CFTs.
Section \ref{ss:discussion} concludes with the discussion and the implications of our results.

\section{Quantum stabilizer codes over rings}\label{ss:review_stabilizer}

The stabilizer formalism of quantum error-correcting codes was originally introduced in the binary case \cite{Gottesman:1996rt} and it has been generalized to finite fields \cite{knill1996group,knill1996non,rains1999nonbinary} and finite Frobenius rings \cite{klappenecker2012nice,nadella2012stabilizer,guenda2014quantum}.
In this section, we give a brief review of finite Frobenius rings following \cite{Wood2011APPLICATIONSOF,dougherty2017algebraic} and describe the formulation of quantum stabilizer codes over Frobenius rings. We also emphasize the relationship between quantum stabilizer codes and classical codes and introduce Calderbank-Shor-Steane (CSS) codes as a special type of stabilizer code.

Let us define a finite ring of our interest.
A ring is a set $R$ with addition $(+)$ and multiplication $(\cdot)$ satisfying the following conditions: $(a\in R)$
\begin{itemize}
    \item $(R,+)$ is closed, associative, and commutative. 
    \item $R$ contains additive identity $0$ s.t. $a+0 =a$ and inverse elements $-a$ s.t. $a+(-a) = 0$.
    \item $(R,\cdot)$ is closed, associative, and contains multiplicative identity $1$ s.t. $1\cdot a = a$.
    \item Multiplication is distributive under addition.
\end{itemize}
In particular, a ring is called finite if it has a finite number of elements.

Let $R$ be a finite ring and introduce a character of its module.
For a finite abelian group $G$, a character is a group homomorphism $\chi: G\to \BC^\times$ and the set of all characters form a group $\hat{G} = \mathrm{Hom}(G,\BC^\times)$ called a character group of $G$.
Since we can see a finite ring $R$ as an additive abelian group $(R,+)$, we can similarly define a character of a finite left $R$-module $A$.
In this case, $\hat{A}$ consists of characters of an additive group $A$ and it becomes an right $R$-module called a character module: $(\chi r)(a):=\chi(ra)$ for $\chi\in \hat{A}$, $r\in R$, and $a\in A$. In the same way, we can formulate a character of a right $R$-module $A$ such that $(r \chi )(a):=\chi(ar)$.

In particular, for $A = R$, $R$ is a left and right $R$-module and we can define a character of a finite ring $R$. Then, a character $\chi\in\hat{R}$ induces a left and right homomorphism $R\to\hat{R}$ such that $R\ni r\mapsto \chi r\in\hat{R}$ and $R\ni r\mapsto r \chi \in\hat{R}$, respectively.
We call $\chi\in \hat{R}$ a left (resp., right) generating character if $\hat{R} = \{r \chi \in \hat{R} \mid r\in R\}$ (resp., $\hat{R} = \{\chi r\in \hat{R} \mid r\in R\}$).

Frobenius rings are a special type of finite ring, which can be well characterized by the existence of a generating character. It is known that for a finite ring $R$, the followings are equivalent (\cite{hirano1997admissible,wood1999duality}):
\begin{itemize}
    \item $R$ is a Frobenius ring.
    \item $R$ admits a left generating character.
    \item $R$ admits a right generating character.
\end{itemize}
Moreover, if these conditions are met, every left (right) generating character is also a right (left) generating character.
Therefore, a Frobenius ring $R$ always has a generating character.

We give three examples of finite Frobenius rings and their generating characters.
\begin{enumerate}
    \item For a finite field of order $q=p^m$ denoted by $\BF_q$, a generating character is given by $\chi(a) = \exp(2\pi\i \,\tr (a)/p)$ where $a\in\BF_q$, and $\tr$ is the trace map $\BF_q\to \BF_p$.
    \item For a ring of integers modulo $q$ denoted by $\BZ_q$, a generating character is $\chi(a) = \exp(2\pi\i\, a/q)$ where $a\in\BZ_q$.
    \item For a direct sum of Frobenius rings $R_i$ $(i=1,2,\cdots,n)$ with a generating character $\chi_i$, a generating character is $\chi = \prod_{i=1}^n\chi_i$.
\end{enumerate}

Let $R$ be a finite Frobenius ring with $q$ elements and $B= \{\,\ket{x}\,|\, x\in R\,\}$ be an orthonormal basis of $\BC^q$.
For $a, b\in R$, Pauli-like operators $X(a), Z(b)$ are defined by
\begin{align}
    X(a) \ket{x} = \ket{x + a} \ , \qquad Z(b)\ket{x} = \chi(bx)\,\ket{x} \ ,
\end{align}
where $\chi\in \mathrm{Hom}(R,\BC^\times)$ is a character of the additive abelian group $(R, +)$.
For each character $\chi$,
we have a unique function $\psi\in\mathrm{Hom}(R,\BQ/\BZ)$ such that
\begin{align}
    \chi(x) = e^{2\pi\i\,\psi(x)} \ .
\end{align}
Then, a set of linear operators acting on $\BC^q$ is given by
\begin{align}
    \CE = \left\{\, X(a)\,Z(b)\, |\, a, b \in R\,\right\} \ .
\end{align}
A Frobenius ring $R$ guarantees that the set $\CE$ forms an orthonormal basis with respect to the normalized Hilbert-Schmidt inner product $\langle A|B\rangle = \frac{1}{q}\tr(A^\dag B)$ \cite{klappenecker2012nice}. Here $\tr$ denotes the trace of a matrix and $A^\dag$ is the adjoint of the operator $A$.

We can easily generalize this representation to the $n$ quantum system. An orthonormal basis of $\BC^{q^n} = \BC^q\otimes\cdots\otimes\BC^q$ is the $n$-fold tensor products of $B = \{\,\ket{x}\mid x\in R\,\}$.
Let $\Ba = (a_1, \cdots, a_n)\in R^n$, $\Bb = (b_1, \cdots, b_n)\in R^n$, and define the linear operators
\begin{align}
    \begin{aligned}
        X(\Ba) \equiv X(a_1)\otimes \cdots \otimes X(a_n) \ , \qquad
        Z(\Bb) \equiv Z(b_1)\otimes \cdots \otimes Z(b_n) \ .
    \end{aligned}
\end{align}
We also define the error group by
\begin{align}
    G_n = \left\{\,\omega^\kappa\,X(\Ba)\,Z(\Bb)\,|\, \Ba, \Bb \in R^n, \kappa \in \BZ \, \right\}\ ,
\end{align}
where $\omega = \exp{(2\pi \i/p)}$ and $p$ the exponent of the additive group $(R, +)$.

To simplify the notation, we denote an error operator $g(\Bu)$ by
\begin{align}
    g(\Bu) \equiv X(\Ba)\,Z(\Bb) \qquad \text{for}\quad \Bu = (\Ba\,|\, \Bb)  \in R^{2n} \ .
\end{align}
For $\Bu = (\Ba\,|\, \Bb) \in R^{2n}$ and $\Bv = (\Ba'\,|\, \Bb') \in R^{2n}$,
the product of two error operators satisfy $g(\Bu)\, g(\Bv) = e^{2\pi\i\,\psi(\Bb\cdot \Ba')}\,g(\Bu + \Bv)$ and
\begin{align}
\label{two_error_prod}
    g(\Bu)\, g(\Bv) = e^{2\pi\i\,\psi(\Bb\cdot\Ba'-\Bb'\cdot\Ba)}\,g(\Bv)  g(\Bu)\ .
\end{align}
Then, the two error operators commute if and only if $\psi(\Bb\cdot\Ba'-\Bb'\cdot\Ba) = 0$.

For an abelian subgroup $S$ of $G_n$, a quantum stabilizer code $V_S$ is defined by a subspace of $\BC^{q^n}$ fixed by every element in the abelian subgroup $S$
\begin{align}
    V_S = \left\{\, \ket{\psi} \in \BC^{q^n}\;\middle|\; g\,\ket{\psi} = \ket{\psi}~\text{for}~\forall g\in S\,\right\} \ .
\end{align}
Here, the abelian group $S$ is called the stabilizer group of $V_S$.
The projector on code subspace $V_S$ is given by
\begin{align}
    P_S = \frac{1}{|S|} \sum_{g\,\in\, S} g\ ,
\end{align}
and we can find the dimension of $V_S$ by tracing over the projector. We find from \cite{GHEORGHIU2014505}
\begin{align}
    \dim{V_S} &= K = \frac{1}{|S|}\,q^n \,.
\end{align}

While stabilizer operators leave states in the code subspace $V_S$ invariant, there exist elements of the Pauli group that change one state in the code subspace into another but leave the whole code subspace unchanged.
These operators are not detectable and are called logical operators.
As the code subspace is invariant under the action of logical operators, they commute with any element in the stabilizer group $S$.
Therefore, the set of logical operators can be written as $N(S)\backslash S$ where $N(S)$ is the normalizer of the stabilizer group
\begin{align}
    N(S) = \left\{\,g \in G_n \;\middle|\; g\,s\,g^\dag \in S,\, \forall s \in S \,\right\}\,.
\end{align}
On the other hand, the operators outside the normalizer $N(S)$ do change a state in the code subspace $V_S$ into a state that is not present in $V_S$, so such an action is detectable and called an error operator. Clearly, the set of error operators is given by $G_n\backslash N(S)$. 

Let us characterize the error-correcting property of quantum stabilizer codes.
Defining the weight $\text{wt}(g)$ of an operator $g \in G_n$ is the number of qudits on which it acts non-trivially, the distance $d$ of a stabilizer code is
\begin{align}
    d=\min_{g\,\in\, N(S) \backslash S} \text{wt}(g)\ .
\end{align}
The distance gives a measure of the error-correcting capabilities of the code. A code with distance $d$ can detect errors of weight $d-1$ and correct errors of weight $\lfloor{(d-1)/2}\rfloor$.
Conventionally, a $K$-dimensional subspace of $\BC^{q^n}$ with distance $d$ is called an $((n,K,d))_{q}$ quantum stabilizer code.

We will make use of
an intriguing relationship between quantum stabilizer codes and classical codes~\cite{Calderbank:1996aj,Calderbank:1996hm,knill1996non,nadella2012stabilizer}.
For a finite Frobenius ring $R$, a classical code $C$ is defined as an additive subgroup of $(R^{2n}, +)$.
Associated with a stabilizer group $S\subset G_n$, we have a classical code
\begin{align}
    \CC = \left\{\Bu = (\Ba\,|\, \Bb)\in R^{2n}\;\middle|\;g(\Bu)\in S\,\right\}\,.
\end{align}
As the stabilizer group $S$ forms a group under multiplication, correspondingly, the classical code $\CC\subset  R^{2n}$ is closed under addition.

To establish the relationship between stabilizer codes and classical codes, it is natural to introduce the form $\langle *,*\rangle_\chi :R^{2n}\times R^{2n}\to \BQ/\BZ$
\begin{align}\label{innerprod}
    \langle \Bu, \Bv\rangle_\chi \equiv \psi(\Bb\cdot \Ba' - \Bb'\cdot \Ba) \ , \qquad \text{for}\quad \Bu = (\Ba\,|\, \Bb)\ , \quad \Bv = (\Ba'\,|\, \Bb') \ .
\end{align}
For a Frobenius ring $R$ with a generating character $\chi$, the form $\langle \Bu, \Bv\rangle_\chi$ becomes $\BZ$-bilinear and left- and right-nondegenerate \cite{nadella2012stabilizer}. Then, $\langle *,*\rangle_\chi$ is a symplectic bilinear form.

Associated with the symplectic bilinear form, the dual code $\CC^{\perp_s}$ is defined by\footnote{One can define another dual code by
\begin{align}
    {}^{\perp_s} \CC \equiv \left\{\, \Bu \in R^{2n}\,|\, \langle \Bu, \Bv\rangle_\chi = 0 ~ \text{for}~\forall \Bv\in \CC\,\right\} \ .
\end{align}
}
\begin{align} \label{Cperp}
    \CC^{\perp_s} \equiv \left\{\, \Bu \in R^{2n}\,|\, \langle \Bv, \Bu\rangle_\chi = 0 ~ \text{for}~\forall \Bv\in \CC\,\right\} \ .
\end{align}
For a Frobenius ring $R$ with $q$ elements, the cardinality of the dual code is given by $|\CC^{\perp_s}| = |R^{2n}|/|\CC| = q^{2n}/|\CC|$.
We call $\CC$ self-orthogonal when $\CC \subset \CC^{\perp_s}$.
For a self-orthogonal code $\CC$, any pair of codewords $\Bu, \Bv \in \CC$ is orthogonal, $\langle \Bu, \Bv\rangle_\chi = 0$. We call $\CC$ self-dual if $\CC=\CC^{\perp_s}$. In this case, $|\CC| =|\CC^{\perp_s}| = q^n$.

By the abelian property of a stabilizer group $S$, any pair of codewords $\Bu, \Bv \in \CC$ is orthogonal with respect to the symplectic bilinear form.
Accordingly, a classical code $\CC$ constructed from a stabilizer group $S$ is always self-orthogonal.
In the opposite direction, once a self-orthogonal code is given, one can easily construct a commuting set of operators.
Hence, we have the following proposition.

\begin{proposition}[{\cite[Theorem 7]{nadella2012stabilizer}}]
    Let $R$ be a finite Frobenius ring. A stabilizer group $S\subset G_n$ exists if and only if there exists a classical additive code self-orthogonal with respect to the symplectic bilinear form $\langle *,*\rangle_\chi$.
    \label{prop:additive_stab}
\end{proposition}

There is an important class of quantum stabilizer codes defined by a pair of classical linear codes, which we will heavily use in later sections. 
For a finite Frobenius ring $R$, a left (right) \textit{linear} code $C$ of length $n$ is a left (right) submodule of $R^n$.
We introduce the standard Euclidean inner product $\cdot$ that maps $R^n\times R^n$ into $\BQ/\BZ$.
For a left linear code $C$, the dual code with respect to the Euclidean inner product can be defined as (\cite{wood1999duality})
\begin{align} \label{eq:dual_code_Euclidean}
    C^\perp = \left\{\,\Bb\in R^n\;\middle|\; \Ba\cdot \Bb = 0\,,\;\Ba\in C\,\right\}\,.
\end{align}
In analogy with the symplectic bilinear form $\langle *,*\rangle_\chi$, we call $C$ self-orthogonal if $C\subset C^\perp$ and self-dual if $C=C^\perp$.

Let $C_X$ and $C_Z$ be classical linear codes such that $C_Z^\perp \subset C_X$.
To construct a stabilizer code, let us introduce a classical code $\CC=C_X^\perp\times C_Z^\perp\in R^{2n}$
\begin{align}
\label{eq:css_frobenius}
    \CC = \left\{\,\Bu = (\Bc_x\,|\,\Bc_z)\in R^{2n}\; \middle| \; \Bc_x\in C_X^\perp\,,\;\Bc_z\in C_Z^\perp\,\right\}\,.
\end{align}
For codewords $\Bu=(\Bc_x\,|\,\Bc_z)$, $\Bv=(\Bc_x'\,|\,\Bc_z')\in\CC$, we have
\begin{align}
    \Bc_z\cdot \Bc_x' - \Bc_z'\cdot \Bc_x = 0\,.
\end{align}
Therefore, the symplectic bilinear form vanishes: $\langle \Bu,\Bv\rangle_\chi=0$, so $\CC = C_X^\perp\times C_Z^\perp$ is self-orthogonal with respect to the symplectic bilinear form.
Then, Proposition \ref{prop:additive_stab} tells us that we can construct a stabilizer group
\begin{align}
    S = \left\{\,g(\Bu)\in G_n\mid \Bu = (\Bc_x\,|\,\Bc_z)\in C_X^\perp\times C_Z^\perp\,\right\}\,.
\end{align}
This construction was first given by Calderbank, Shor, and Steane \cite{calderbank1996good,steane1996multiple}, so such a stabilizer code is called CSS type. After the development in the binary case, it was generalized into a finite Frobenius ring \cite{guenda2014quantum}.

\section{Review of Narain CFTs}\label{ss:review_narain}
The purpose of this section is to review the basics of Narain CFTs and their lattice structure~\cite{Narain:1985jj,Narain:1986am}, while summarizing our conventions.
We also review the rationality of Narain CFTs with a small central charge.
The reader familiar with Narain CFTs can skim this section and skip to the next sections.

\subsection{Momentum lattices}
Narain CFTs are two-dimensional theories of $n$ compact bosons $X^\mu(z,\bar z)$ whose target space is the $n$ torus \cite{Narain:1985jj,Narain:1986am}:
\begin{align}
    X^\mu \simeq X^\mu + 2\pi R\ , \qquad \mu = 1,\cdots, n \ .
\end{align}
The geometry of the $n$ torus depends on the metric $G_{\mu\nu}$ and the anti-symmetric tensor $B_{\mu\nu}$.
Let $e_\mu^{~i}$ be a vielbein with tangent space index $i$ ($i=1,\cdots, n$) defined by
\begin{align}\label{vielbein}
    G_{\mu\nu} = e_\mu^{~i}\,e_\nu^{~j}\,\delta_{ij} \ .
\end{align}
We introduce the dimensionless fields $X^{(i)} \equiv \sqrt{\frac{2}{\alpha'}}\,e_\mu^{~i}\,X^\mu$, and decompose them into the left- and right-moving parts,
$X^{(i)}(z, \bar z) = X^{(i)}_L(z) + X^{(i)}_R(\bar z)$.
The left- and right-moving momenta are parametrized by integer vectors $l_\mu, w^\mu \in \BZ$ as \cite{Polchinski:1998rq}
\begin{align}\label{Narain_momentum_vectors}
    \begin{aligned}
        (p_{L})_i
            &=
                \tilde e_i^{~\mu}\,\left[\, \frac{l_\mu}{r} + \frac{r}{2}\,(B+G)_{\mu\nu}\,w^\nu\right]\ ,\\
        (p_{R})_i
            &=
                \tilde e_i^{~\mu}\,\left[ \,\frac{l_\mu}{r} + \frac{r}{2}\,(B-G)_{\mu\nu}\,w^\nu\right]\ .
    \end{aligned}    
\end{align}
Here, $\tilde e_i^{~\mu}$ is the inverse of the vielbein, i.e., $\tilde e_i^{~\mu}e_\mu^{~j} = \delta_i^j$ and $r \equiv R\,\sqrt{\frac{2}{\alpha'}}$ is the dimensionless radius.

Let $\eta_0$ be the inner product for momentum vectors $P = (p_L, p_R)$ defined by 
\begin{align}
    \eta_0 (P, P') \equiv p_L \cdot p_L' - p_R \cdot p_R' \ ,
\end{align}
where $\cdot$ is the standard Euclidean inner product.
The set of the momentum vectors satisfying \eqref{Narain_momentum_vectors} forms a lattice $\Gamma$ in $\BR^{n,n}$.
Narain CFTs are bosonic if $\Gamma$ is an even lattice, i.e., for any vector $P\in\Gamma$,
\begin{align}
    \eta_0(P, P) \in 2\,\BZ \ .
\end{align}
The dual lattice $\Gamma^\ast$ with respect to the metric $\eta_0$ is defined by
\begin{align}
    \Gamma^\ast = \left\{\, P \in \BR^{n,n}\, \bigg| \, ~\eta_0(P, P') = \BZ \ , ~ P'\in \Gamma\, \right\}\ .
\end{align}
Narain CFTs are modular invariant if $\Gamma$ is a self-dual lattice, i.e., $\Gamma = \Gamma^\ast$.
In what follows, we consider Narain CFTs with even self-dual momentum lattices.

To any vector $P\in \Gamma$ in the momentum lattice, we associate the vertex operator by
\begin{align}
    V_{p_L,p_R}(z,\bar{z}) = : e^{\i\,p_L\cdot X_L(z) + \i\,p_R \cdot X_R(\bar{z})} : \ .
\end{align}
The vertex operator corresponds to the momentum state $\ket{\,p_L,p_R\,}$ via the state-operator mapping.
The Hilbert space $\CH$ of a Narain CFT can be built as
\begin{align} \label{eq:Hilbert_space}
    \CH= \left\{ 
    \alpha_{-k_1}^{i_1} \dots \alpha_{-k_r}^{i_r} 
    \Tilde{\alpha}_{-l_1}^{j_1} \dots \Tilde{\alpha}_{-l_s}^{j_s} 
    \ket{\,p_L,p_R\,} \, \bigg| \, (p_L,p_R) \in \Gamma\,
    \right\} \, ,
\end{align}
with $k_1\dots,k_r \in \BZ^+$ and $l_1\dots,l_s \in \BZ^+$, and $\alpha_k^i$ and $\Tilde{\alpha}_k^i \, (i=1,...,n)$ satisfy the algebra:
\begin{align}
    [\alpha_k^i,\alpha_l^j] = [\Tilde{\alpha}_k^i,\Tilde{\alpha}_l^j] = k\,\delta_{k+l,0}\,\delta^{i,j}, \qquad k,\,l \in \BZ \, .
\end{align}

The torus partition function of the Narain CFT with a momentum lattice $\Gamma$ is given as follows:
\begin{align}\label{torus_PF}
    \begin{aligned}
        Z(\tau,\bar{\tau}) 
            &= 
                \Tr_{\CH} \left[ q^{L_0-\frac{n}{24}} \bar{q}^{\bar{L}_0-\frac{n}{24}}\right] \\    
            &=
            \frac{1}{|\eta(\tau)|^{2n}}\sum_{P \,\in\, \Gamma} q^{\frac{p_L^2}{2}} \bar{q}^{\frac{p_R^2}{2}} \\
            &= 
                \frac{\Theta_{\Gamma}(\tau,\bar{\tau})}{|\eta(\tau)|^{2n}} \ ,        
    \end{aligned}
\end{align}
where $\eta(\tau)$ is the Dedekind eta function, $\Theta_{\Gamma}(\tau,\bar{\tau})$ is the lattice theta function of momentum lattice $\Gamma$, and $q=e^{2\pi \i \tau}$ with $\tau = \tau_1+\i\,\tau_2$.

\subsection{Rationality for theories with small central charges}
Rational CFTs are a special class of CFTs that have a finite number of representations of some chiral algebras \cite{Friedan:1983xq,Moore:1988uz,Moore:1988qv,Moore:1989yh}.
Here, we make some comments on the rationality of Narain CFTs with small central charges.

Narain CFTs with central charge $c=n=1$ are described by a free compact boson $X^{\mu=1}$. 
Switching to the dimensional field $X^{(1)}$, the periodicity of the boson becomes
\begin{align}\label{radius_dimensionless_boson}
    X^{(1)} \simeq X^{(1)} + 2\pi R^{(1)} \ , \qquad R^{(1)} \equiv r\,\sqrt{G_{11}} \ .
\end{align}
The theory is known to be rational when $R^{(1)} = \sqrt{\frac{2\,s}{t}}$ for a pair of coprime positive integers $s, t$ \cite{DiFrancesco:1997nk}.

Narain CFTs with central charge $c=n=2$ consist of two compact bosons denoted by $x = X^{\mu=1}$ and $y = X^{\mu =2}$.
In this case, the rationality is related to the four real parameters parametrizing the moduli space:
\begin{align}
    G = \begin{bmatrix}
        G_{11} & G_{12} \\
        G_{12} & G_{22}
    \end{bmatrix} \quad \text{and} \quad 
    B = \begin{bmatrix}
        0 & B_{12} \\
        -B_{12} & 0
    \end{bmatrix} \, .
\end{align}
However, we can rewrite these moduli parameters in terms of the complex structure modulus $\bm{\tau}$ and the complexified K\"{a}hler modulus $\bm{\rho}$ defined by (see e.g., \cite{Polchinski:1998rq})
\begin{align}
\label{eq:gen_tau_rho}
    \bm{\tau} &= \bm{\tau}_1 + \i\, \bm{\tau}_2 = \frac{G_{12}}{G_{11}} + \i\, \frac{\sqrt{\det G}}{G_{11}} \\
    \bm{\rho} &= \bm{\rho}_1 + \i\, \bm{\rho}_2 = \frac{R^2}{\alpha'} (B_{12} + \i\, \sqrt{\det G}) \,.
\end{align}
Then, it is straightforward to check that
\begin{align}
\begin{aligned}
    \d s^2 &= \frac{\alpha'\,\bm{\rho}_2}{R^2\,\bm{\tau}_2} \,\left|\d x+\bm{\tau} \,\d y\right|^2\\
    &= G_{11} \,\d x^2 + 2G_{12} \,\d x\d y + G_{22}\, \d y^2\,.
\end{aligned}
\end{align}
Note that, in the dimensionless radius $r = R \sqrt{2/\alpha'}$, we obtain
\begin{align}
\begin{aligned}
    \bm{\rho} = \frac{r^2}{2} (B_{12} + \i\, \sqrt{\det G})\,.
\end{aligned}
\end{align}
For a rational Narain theory, the complex parameter $\bm{\tau}$ satisfies 
\begin{align} \label{eq:tau}
    a\, \bm{\tau}^2 + b\, \bm{\tau} + c = 0 \ ,
\end{align} 
with relatively prime coefficients $a,b$ and $c$ and is an element of the \textit{imaginary quadratic number field} (meaning the discriminant $D<0$, which is the case in all relevant applications in physics):
\begin{align}
    K = \BQ(\sqrt{D}) \, ,
\end{align}
$\bm{\tau} \in K$, where $D=b^2-4\,a\,c$. This means $\bm{\tau} = \alpha + \beta\, \sqrt{D}$ where $\alpha,\beta \in \BQ$ and is a modulus of the elliptic curve
\begin{align}
    E = \BC / (\BZ \oplus \bm{\tau}\, \BZ) \,.
\end{align}
We also obtain the same relations to the other complex parameter $\bm{\rho}$. 
One result of \cite{Gukov:2002nw} is that a Narain CFT is rational \textit{if and only if} both $\bm{\tau}$ and $\bm{\rho}$ take values in the same imaginary quadratic number field:
\begin{align}
    \text{Rational CFT} \quad \Longleftrightarrow \quad \bm{\tau},\,\bm{\rho} \in \BQ(\sqrt{D}) \, .
\end{align} 
In our constructions, up to absorption of rational coefficients into $\beta$ for $\bm{\tau}\in K$, $D=-1$.

\section{Narain CFTs from quantum stabilizer codes over $\BZ_k$} 
\label{ss:ringcodes}

In this section, we construct Narain CFTs from quantum stabilizer codes over a finite ring $\BZ_k = \{0,1,2,\cdots,k-1\}$ where $k$ is a positive integer.
Our strategy is to extend the previous construction~\cite{Dymarsky:2020qom, Kawabata:2022jxt}.
In section~\ref{ss:ring_lattice}, we use the relationship between quantum stabilizer codes and classical codes shown in section~\ref{ss:review_stabilizer} to provide a way of constructing Lorentzian even self-dual lattices.
Then, regarding the resulting lattice as a momentum lattice of a Narain CFT, we obtain the torus partition function of the whole Hilbert space for a Narain code CFT in section~\ref{ss:ring_narain}.
After section~\ref{ss:ring_example}, we show various examples of Narain code CFTs with small central charges, which allows us to identify the corresponding modulus and discuss the rational structure of Narain code CFTs.

\subsection{Lorentzian lattices via Construction A}
\label{ss:ring_lattice}

In this section, we focus on a finite ring $\BZ_k = \{0,1,2,\cdots,k-1\}$ of integers modulo a positive integer $k$.
As the finite ring $\BZ_k$ is a special type of finite Frobenius ring, we can formulate a quantum stabilizer code over $\BZ_k$ as in section~\ref{ss:review_stabilizer}.
This subsection is devoted to constructing Lorentzian even self-dual lattices from those quantum stabilizer codes. 

Let $S\subset G_n$ be an abelian group.
Then, we obtain a quantum stabilizer code as a fixed subspace of $\BC^{k^n}$ by any element of the stabilizer group $S$.
From Proposition \ref{prop:additive_stab},
we have the corresponding classical code $\CC\subseteq \BZ_k^{2n}$ 
\begin{align}
    \CC = \left\{\Bu = (\Ba\,|\, \Bb)\in \BZ_k^{2n}\;\middle|\;g(\Bu)\in S\,\right\}\,.
\end{align}

We introduce the off-diagonal Lorentzian inner product on $\BZ_k^{2n}$. For $\Bu = (\Ba,\Bb)$, $\Bv = (\Ba',\Bb')\in\BZ_k^{2n}$, we define the inner product by
\begin{align}
\label{eq:lor-inner}
    \eta(\Bu,\Bv) = \Bu \,\eta\, \Bv^T = \Ba\cdot\Bb' + \Ba'\cdot\Bb\,,\qquad 
    \eta = \begin{bmatrix}
        ~0~ & I_n~ \\
        ~I_n~ & 0~
    \end{bmatrix}\,.
\end{align}
Associated with the off-diagonal Lorentzian inner product, we define the dual code $\CC^\perp$ as
\begin{align}
    \label{eq:CC_perp}
    \CC^\perp = \left\{\, \Bv \in \BZ_k^{2n} \; \middle|\;  \eta(\Bu,\Bv) = 0 \;\mathrm{mod}\; k \, , \; \Bu \in \CC \,\right\} \, . 
\end{align}
We call $\CC$ self-orthogonal if $\CC\subset\CC^\perp$ and self-dual if $\CC = \CC^\perp$.
Also, we call $\CC$ even if $\eta(\Bu,\Bu) \in 2k\,\BZ$ for any $\Bu\in \CC$.

Note that self-orthogonality and evenness are correlated.
Indeed, if $\CC$ is even, then $\CC$ is self-orthogonal.
To see this, suppose that $\CC$ is an even code with respect to $\eta$.
Then, the norm of $\Bu+\Bv\in\CC$ for $\Bu$, $\Bv\in\CC$ is
\begin{align}
    \eta(\Bu+\Bv,\Bu+\Bv) = \eta(\Bu,\Bu) + \eta(\Bv,\Bv) + 2\,\eta(\Bu,\Bv)\,.
\end{align}
As all the norms are in $2k\,\BZ$, then we know $\eta(\Bu,\Bv)\in k\,\BZ$, which implies the self-orthogonality of $\CC$.
For an odd $k\in2\BZ+1$, the opposite is also true.
For $k\in2\BZ+1$, if $\CC$ is self-orthogonal, then $\CC$ is even.
To check this, consider the norm $\eta(\Bu,\Bu) = 2\Ba\cdot\Bb\in 2\BZ$ where $\Bu = (\Ba,\Bb)\in\CC$.
Combining the assumption of self-orthogonality, the norm $\eta(\Bu,\Bu)$ is even and a multiple of $k$.
Hence, for an odd $k$, we conclude $\eta(\Bu,\Bu)\in 2k\,\BZ$ (evenness).

The classical code $\CC$ provides the associated lattice via Construction A \cite{leech1971sphere,conway2013sphere}.
Explicitly, we obtain the Construction A lattice $\Lambda(\CC)$
\begin{align} 
\label{eq:const-A}
\Lambda(\CC) = 
\bigg\{
\left( 
\frac{ \Ba + k\, m_1}{\sqrt{k}}
,
\frac{ \Bb + k\, m_2}{\sqrt{k}}
\right)
\in \BR^n \times \BR^n
\;\bigg|\;
m_1,m_2\in\BZ^n \,,
\
(\Ba,\Bb) \in \CC
\bigg\} \, .
\end{align}
The Construction A lattice $\Lambda(\CC)$ has the off-diagonal Lorentzian inner product \eqref{eq:lor-inner}.
In analogy with classical codes, we define the dual lattice $\Lambda^*$ of $\Lambda$ by
\begin{align}
    \Lambda^* = \left\{\lambda'\in\BR^{2n}\;\middle|\;\eta(\lambda,\lambda') \in\BZ\,,\;\lambda\in\Lambda \right\}\,.
\end{align}
The lattice $\Lambda$ is called self-dual if $\Lambda = \Lambda^*$.
Also, we call $\Lambda$ even if it satisfies $\eta(\lambda,\lambda)\in 2\BZ$ for any element $\lambda\in\Lambda$.
Construction A relates the self-duality of $\CC$ to the one of $\Lambda(\CC)$ as in the following proposition.

\begin{proposition}
    Let $\CC$ be a classical code over $\BZ_k$.
    Then, the Construction A lattice $\Lambda(\CC)$ is self-dual with respect to $\eta$ if and only if $\CC$ is self-dual with respect to $\eta$.  
    \label{prop:ring_self}
\end{proposition}
\begin{proof}
    We start by proving $\Lambda(\CC)^* = \Lambda(\CC^\perp)$.
    Let us take lattice vectors $\lambda = (\lambda_1,\lambda_2)$, $\lambda' = (\lambda_1',\lambda_2')$ such that
    \begin{align}
     \lambda_1 = \frac{ \Ba + k\, m_1}{\sqrt{k}} \ , &\qquad \lambda_2 = \frac{ \Bb + k\, m_2}{\sqrt{k}} \ ,\label{eq:latvec}\\
     \lambda_1' = \frac{ \Ba' + k\, m_1'}{\sqrt{k}} \ , &\qquad \lambda_2' = \frac{ \Bb' + k\, m_2'}{\sqrt{k}} \ ,\label{eq:latvecp}
    \end{align}
    where $\Bu = (\Ba,\Bb), \, \Bv = (\Ba',\Bb')$ and $ m_1,m_2 \in \BZ^n$. Then, the inner product between them is given by
    \begin{align}
        \eta(\lambda_,\lambda') =   \frac{\eta(\Bu,\Bv)}{k} + (\Ba \cdot m_2' + \Ba' \cdot m_2 + \Bb \cdot m_1' + \Bb' \cdot m_1) + k\,(m_1 \cdot m_2' + m_1' \cdot m_2) \,.\label{eq:latvecprod}
    \end{align}
    
    Suppose $\lambda\in\Lambda(\CC)$ and $\lambda' \in \Lambda(\CC^\perp)$, then we can write $\lambda'$ as \eqref{eq:latvecp} with $\Bv \in \CC^\perp$. On the other hand, any $\lambda \in \Lambda(\CC)$ can be written as \eqref{eq:latvec} with $\Bu \in \CC$. 
    As $\eta(\Bu,\Bv) = 0  \in k\,\BZ$ for $\Bu\in\CC$ and $\Bv\in\CC^\perp$, we obtain $\eta(\lambda,\lambda')\in\BZ$, which concludes
    $\Lambda(\CC^\perp) \subset \Lambda(\CC)^*$.

    Suppose $\lambda = (\lambda_1,\lambda_2)\in\Lambda(\CC)$ and $\lambda' = (\lambda_1',\lambda_2')\in\Lambda(\CC)^*$.
    Let us take $\lambda_1 = \sqrt{k}\, m_1$ and $\lambda_2 = \sqrt{k}\, m_2$.
    Then, the inner product \eqref{eq:latvecprod} becomes
    \begin{align}
       {\eta}(\lambda_,\lambda') = \sqrt{k}\, m_1 \cdot \lambda_2' + \sqrt{k}\, m_2 \cdot \lambda_1'\ .
    \end{align}
    As $\eta(\lambda,\lambda')\in\BZ$, we see that $\lambda' \in (\BZ/\sqrt{k})^n$ and may write the lattice vector $\lambda' \in \Lambda(\CC)^*$ in the form of \eqref{eq:latvecp} where $\Bv = (\Ba',\Bb')\in \BZ_k^{2n}$. 
    For the inner product $\eta(\Bu,\Bv)$ to be an integer for arbitrary $\Bu\in\CC$, it is implied $\Bv \in \CC^\perp$ and $\lambda' \in \Lambda(\CC^\perp)$. We conclude $\Lambda(\CC)^* \subset \Lambda(\CC^\perp)$.  

    Thus, the lattice $\Lambda(\CC^\perp)$ is the dual of lattice $\Lambda(\CC)$\,: $\Lambda(\CC)^* = \Lambda(\CC^\perp)$. For a self-dual code, $\CC = \CC^\perp$, the lattice is also self-dual $\Lambda(\CC)^* = \Lambda(\CC)$. As lattices and codes have a one-to-one correspondence, $\Lambda(\CC)=\Lambda(\CC')$ if and only if $\CC = \CC'$. Therefore, $\Lambda(\CC)$ is self-dual with respect to ${\eta}$ if and only if $\CC$ is self-dual with respect to $\eta$.
\end{proof}

Additionally, the evenness of a classical code $\CC$ and the Construction A lattice $\Lambda(\CC)$ are equivalent.

\begin{proposition}
    The Construction A lattice $\Lambda(\CC)$ is even with respect to ${\eta}$ if and only if a classical code $\CC$ is even with respect to $\eta$.
    \label{prop:ring_even}
\end{proposition}
\begin{proof}
    Let $\lambda\in\Lambda(\CC)$ be an element denoted by \eqref{eq:latvec}.
    From the relation 
    \begin{align}
        {\eta}(\lambda_,\lambda) =   \frac{\eta(\Bu,\Bu)}{k} + 2\,(\Ba \cdot m_2+ \Bb \cdot m_1) + 2\,k\, m_1 \cdot m_2\,,
    \end{align}
    we see that the second and third terms are even. 
    Assuming an even code $\CC$, the norm is $\eta(\Bu,\Bu) \in 2k\, \BZ$ which gives an even lattice $\Lambda(\CC)$. 
    If the lattice $\Lambda(\CC)$ is even, this implies that $\CC$ is even. 
\end{proof}

Combining Proposition~\ref{prop:ring_self} and \ref{prop:ring_even}, we arrive at the following theorem essential to the construction of Narain CFTs.
\begin{theorem}
    The Construction A lattice $\Lambda(\CC)$ is even self-dual with respect to $\eta$ if and only if a classical code $\CC$ is even self-dual with respect to $\eta$.
    \label{theorem:self_duality_ring_lattice}
\end{theorem}

To construct a Lorentzian even self-dual lattice, we need to prepare an even self-dual code $\CC$. We can use the CSS construction of quantum stabilizer codes introduced in \eqref{eq:css_frobenius}.
Let $C$ be a classical linear code over $\BZ_k$. Using the dual code $C^\perp$ with respect to the Euclidean inner product, we set $C_X = C^\perp$ and $C_Z = C$. Then, the condition $C_Z^\perp\subset C_X$ is satisfied. Hence, we obtain the classical code $\CC = C \times C^\perp\subset\BZ_k^{2n}$ 
\begin{align}
    \label{eq:css-stab}
    \CC = \left\{\, \Bu =(\Ba,\Bb) \in \BZ^{2n}_k \; \middle|\;  \Ba \in C, \, \Bb \in C^\perp \right\} \, .
\end{align}
These classical codes $\CC$ specify the CSS-type of quantum stabilizer codes.
We can show that the classical code $\CC$ yields a Lorentzian even self-dual lattice via Construction A.

\begin{proposition}
    Let $C$ be a classical linear code over $\BZ_k$ and $\CC$ the classical code given by the CSS construction \eqref{eq:css-stab}.
    Then, the Construction A lattice $\Lambda(\CC)$ is even self-dual with respect to $\eta$.
    \label{prop:CSS_ring}
\end{proposition}
\begin{proof}
    From Theorem~\ref{theorem:self_duality_ring_lattice}, we verify the evenness and self-duality of the classical code $\CC = C\times C^\perp$.
    To see evenness, consider the norm $\eta(\Bu,\Bu) = 2\,\Ba\cdot \Bb$ where $\Bu = (\Ba,\Bb)\in C\times C^\perp$. The norm is in $2\,k\,\BZ$ because $\Ba\cdot\Bb \in k\,\BZ$ for $\Ba\in C$ and $\Bb\in C^\perp$.
    Hence, the CSS construction $\CC$ is even with respect to $\eta$.
    To see self-duality, we expand (\ref{eq:CC_perp}) as
    \begin{align}
        \label{eq:CC_perp_exp}
        \CC^\perp = \left\{\, (\Ba',\Bb') \in \BZ^{2n}_k \;\middle|\;  \Ba \cdot \Bb' + \Bb \cdot \Ba' = 0\;\,\mathrm{mod} \;k\,, \; \Ba \in C\,, \;  \Bb \in C^{\perp}\, \right\}\,.
    \end{align}
    Then, by independently considering the cases $\Ba = 0 \in C$ and $\Bb = 0 \in C^\perp$, we obtain the requirements $\Bb \cdot \Ba' = 0\mod k$ for $\Bb\in C^\perp$ and $\Ba \cdot \Bb' = 0\mod k$  for $\Ba\in C$, respectively. 
    These conditions reduce to $\Ba'\in C$ and $\Bb'\in C^\perp$.
    Therefore, we have
    \begin{align}
        \CC^\perp\subset \left\{\, (\Ba',\Bb') \in \BZ^{2n}_k \, |\,  \Ba' \in C, \, \Bb' \in C^\perp \right\} \equiv \CC \, . 
    \end{align}
    On the other hand, it is straightforward to see $\CC\subset \CC^\perp$, then we conclude $\CC=\CC^\perp$.
\end{proof}

\subsection{Narain code CFTs}
\label{ss:ring_narain}

By identifying the even self-dual code lattice $\Lambda(\CC)$ from Construction A with a momentum lattice $\Gamma$, we can construct a Narain code CFT associated with the code $\CC$.
The lattice vectors $\lambda \in \Lambda(\CC)$ are related to the left-moving and right-moving momenta $(p_L,p_R)$ as follows:
\begin{align}
    \label{eq:frame-relation}
    (\lambda_1,\lambda_2) = \left( \frac{p_L + p_R}{\sqrt{2}}, \frac{p_L - p_R}{\sqrt{2}}\right) \in \Lambda(\CC) \, .
\end{align}
We denote the set of elements $P = (p_L, p_R)$ related to $(\lambda_1,\lambda_2)\in\Lambda(\CC)$ by \eqref{eq:frame-relation} as $\Tilde{\Lambda}(\CC)$.
In the left- and right-moving momenta frame, the norm of $\lambda$ with respect to ${\eta}$ is 
\begin{align}
    {\eta}(\lambda,\lambda) = p_L^2 - p_R^2 = \eta_0(P,P)\ .
\end{align}

The torus partition function \eqref{torus_PF} of the Narain code CFT can be obtained by setting $\Gamma = \tilde\Lambda(\CC)$\,:
\begin{align}
    Z_\CC(\tau,\bar{\tau}) = \frac{1}{|\eta(\tau)|^{2n}} \sum_{(\Ba,\Bb)\,\in\, \CC} \;\sum_{m_1,m_2\, \in\, \BZ^n} q^{\frac{k}{4} \left(\frac{\Ba+\Bb}{k} + m_1+m_2\right)^2} \bar{q}^{\frac{k}{4} \left(\frac{\Ba-\Bb}{k} + m_1-m_2\right)^2} \, .
\end{align}
As in the previous constructions \cite{Dymarsky:2020qom,Kawabata:2022jxt} of Narain code CFTs, we can express these partition functions in terms of the complete weight enumerators\footnote{They are also called complete enumerator polynomials. We will use the two terms interchangeably in this paper.}
of classical codes $\CC$.
We define the complete weight enumerator of $\CC$ by
\begin{align}
    W_\CC(\{x_{ab}\}) = \sum_{c \,\in\, \CC} \;\prod_{(a,b)\,\in\, \BZ_k \times \BZ_k} x_{ab}^{\text{wt}_{ab}(c)}\ ,
\end{align}
where $x_{ab}$ are variables for $a, b\in \BZ_k$ and $\text{wt}_{ab}$ is defined as the number of components $c_i = (\Ba_i, \Bb_i) \in \BZ_k \times \BZ_k$ equal to $(a, b) \in \BZ_k \times \BZ_k$ for a codeword in $\CC$\,:
\begin{align}
    \text{wt}_{ab}(c) = |\, \{i \, |\, c_i = (a,b)\}\, |\,.
\end{align}

Then, the partition function of the Narain code CFT can be uniquely determined by the complete weight enumerator of $\CC$ as
\begin{align}\label{eq:ZC-WC}
     Z_\CC(\tau,\bar{\tau}) = \frac{1}{|\eta(\tau)|^{2n}}\,  W_\CC(\{\psi_{ab}\}) \ ,
\end{align}
where the functions $\psi_{ab}(\tau,\Bar{\tau})$ are given by 
\begin{align} \label{eq:psi_ab}
    \psi_{ab}(\tau,\bar{\tau}) = \sum_{m_1,m_2 \,\in\, \BZ^n} q^{\frac{k}{4} \left(\frac{a+b}{k} + m_1+m_2\right)^2} \bar{q}^{\frac{k}{4} \left(\frac{a-b}{k} + m_1-m_2\right)^2} \, .
\end{align}
It is convenient to express the function $\psi_{ab}$ in terms of the theta function as 
\begin{align}\label{eq:psi-ab-Theta}
    \psi_{ab}(\tau,\bar{\tau}) 
        = 
        \Theta_{a+b,k}(\tau)\, \bar{\Theta}_{a-b,k}(\bar{\tau}) +  \Theta_{a+b-k,k}(\tau) \, \bar{\Theta}_{a-b-k,k}(\bar{\tau})\ ,
\end{align}
where $\Theta_{m,k}(\tau)$ is given by
\begin{align}\label{eq:Theta-m-k}
    \Theta_{m,k}(\tau) = \sum_{n\,\in\,\BZ} q^{k\left(n+\frac{m}{2k}\right)^2} \, .
\end{align}
In particular, for $k=2$, the functions $\psi_{ab}$ can also be expressed as
\begin{align}\label{eq:psi-ab-theta}
\begin{aligned}
    \psi_{00} 
        &= 
        \frac{1}{2} \left( \theta_3(\tau)\,\bar{\theta}_3(\bar{\tau}) + \theta_4(\tau)\,\bar{\theta}_4(\bar{\tau}) \right) \ , \\
    \psi_{11} 
        &= 
        \frac{1}{2} \left( \theta_3(\tau)\,\bar{\theta}_3(\bar{\tau}) - \theta_4(\tau)\,\bar{\theta}_4(\bar{\tau})\right) \ , \\
    \psi_{01} 
        &= 
        \psi_{10} = \frac{1}{2}\, \theta_2(\tau)\,\bar{\theta}_2(\bar{\tau})\ ,
\end{aligned}
\end{align}
where $\theta_i(\tau)$ $(i=2,3,4)$ are the Jacobi theta functions at $z=0$.
We note that the combination $\Theta_{m,k}(\tau)/\eta(\tau)$ is a character of the chiral algebra corresponding to the level-$2k$ Chern-Simons theory $\U(1)_{2k}$.
Thus the partition function~(\ref{eq:ZC-WC}) is a finite sum of products of characters, and the CFT is rational.

When $\mathcal{C}$ is given by the CSS construction, we can read off the metric and the B-field from the lattice as follows.\footnote{
The following is valid not only for codes over $\BZ_k$, but also for codes over $\BF_{p^m}$.}
The Construction A lattice $\Lambda(\CC)$ for a CSS-type code $\CC = C\times C^\perp$ has a product structure: $\Lambda(\CC) = \Lambda(C\times C^\perp) = \Lambda(C)\times\Lambda(C^\perp)$. 
In general, if a matrix $G_{\Lambda}$ is a generator matrix of a lattice $\Lambda$, i.e.,
\begin{align}
    \Lambda = \left\{ \, v \, G_\Lambda \in \BR^n\mid v\in\BZ^n\,\right\}\,,
\end{align}
then $G_\Lambda^\ast := (G_\Lambda^{-1})^T$ is that of the dual lattice $\Lambda^\ast$\,: $G_{\Lambda^\ast} = G_\Lambda^\ast$. Thus, if an $n\times n$ matrix $G_{\Lambda(C)}$ is the generator matrix of $\Lambda(C)$, then the generator matrix of the lattice $\Lambda(\CC)$ can be chosen as
\begin{equation} \label{eq:gen_mat_CC}
    G_{\Lambda(\CC)} = 
    \begin{bmatrix} 
        ~G_{\Lambda(C)}~ & 0 ~\\ 
        ~0~ & G_{\Lambda(C)}^\ast ~
    \end{bmatrix} \,.
\end{equation}
We find this to be a special case of the general generator matrix given by~\cite{Yahagi:2022idq}
\begin{equation}
    \begin{bmatrix}
        \frac{\sqrt{2}}{r}\, \gamma^\ast & 0 \\
        -\frac{r}{\sqrt{2}}\, B\, \gamma^\ast & \frac{r}{\sqrt{2}}\, \gamma
    \end{bmatrix} \ ,
\end{equation}
where $\gamma_{\mu i}=e_\mu^{~i}$ is the vielbein defined in \eqref{vielbein} regarded as a matrix. 
Thus, the Construction A lattice with the generator matrix \eqref{eq:gen_mat_CC} describes a CFT with $\gamma=\frac{\sqrt{2}}{r}\, G_{\Lambda(C)}^\ast$ and $B=0$.
Namely, the metric and anti-symmetric tensor of the Narain code CFT can be chosen as 
\begin{equation}\label{G_B_from_lattice}
    G_{\mu\nu} 
        = 
        \frac{2}{r^2}\,\left(G_{\Lambda(C)}^\ast \,G_{\Lambda(C)}^{-1}\right)_{\mu\nu} \,,\qquad B_{\mu\nu}=0 \,.
\end{equation}
Note that via unimodular transformation on the generator matrix, we can always switch to equivalent points in the Narain moduli space which does have a non-trivial $B$ field.
However, the representative \eqref{G_B_from_lattice} is useful to verify the rationality of our Narain code CFTs with the complex parameters \eqref{eq:gen_tau_rho}.

\subsection{Examples}
\label{ss:ring_example}
We have formulated the construction of Narain CFTs from stabilizer codes over $\BZ_k$.
Here we give a number of examples of CFTs constructed over various finite rings with a small central charge. 
To elucidate the profile of the resultant Narain CFTs, we give a way of computing Narain moduli from their lattice structure.
In particular, for the central charge one $(n=1)$, we see that all rational points of a compact boson emerge from our construction.
For $n=2$, we obtain the complex parameters on the target torus to discuss the rationality.

\subsubsection{Codes over $\BZ_k$ with $n=1$}
Perhaps the simplest nontrivial example of a nontrivial ring is given by $\BZ_4 = \{0,1,2,3\}$. We choose a classical self-dual code $C=C^\perp$ 
\begin{align}
    C = \left\{(0), (2) \right\} \, .
\end{align}
The classical code $\CC$ will then be given by the code words
\begin{align}
    \CC &= \left\{\, (\Ba,\Bb) \in \BZ^2_4 \, |\,  \Ba \in C, \, \Bb \in C \right\} \\
    &= \left\{ (0,0), (0,2), (2,0), (2,2) \right\} \ ,
\end{align}
which we can use to construct the lattice $\Lambda(\CC)$ given by (\ref{eq:const-A}):
\begin{align}
    \Lambda(\CC) 
    &:= 
        \bigg\{
        \left( 
        \frac{ \Ba + 4\, m_1}{2}
        ,
        \frac{ \Bb + 4\, m_2}{2}
        \right)
        \,\bigg|\,
        m_1,m_2\in\BZ \,,
        \
        (\Ba,\Bb) \in \CC
        \bigg\}  
    = 
    \BZ^2 \ ,
\end{align}
which is just equivalent to the square lattice $\BZ^2$. From here we can construct the partition function using the complete weight enumerator.

\paragraph{$\BZ_{k^2}$ with $k\in \BZ$.}

Let us consider a classical code generated by $k$
\begin{align}
    C = \{ (0), (k), (2\,k), \cdots, (k^2 - k) \} \ .
\end{align}
The classical code $\CC$ becomes 
\begin{align}
    \CC &= \left\{\, (\Ba,\Bb) \in \BZ^2_{k^2} \, |\,  \Ba \in C, \, \Bb \in C \right\} \ .
\end{align}
The Construction A lattice is
\begin{align}
    \Lambda(\CC) &:= 
    \bigg\{
    \left( 
    \frac{ \Ba + k^2\, l_1}{k}
    ,
    \frac{ \Bb + k^2\, l_2}{k}
    \right)
    \,\bigg|\,
    l_1,l_2\in\BZ \,,
    \
    (\Ba,\Bb) \in \CC\,
    \bigg\} = \BZ^2\ .
\end{align}
We see that for arbitrary $k^2$, the resulting lattices, and thus CFT, are all the same. 

\paragraph{$\BZ_k$ with $k=m_1m_2$.}
We may also construct CFTs from initial classical codes that are not self-dual. Consider a classical code over $\BZ_6$ generated by $3$. 
Then, the dual code with respect to the Euclidean metric is also generated by $2$ and the codewords are
\begin{align}
    C = \{ (0), (3) \}\ , \qquad  
    C^\perp = \{(0), (2), (4) \} \ .
\end{align}
The classical code $\CC$ becomes 
\begin{align}
    \CC &= \left\{\, (0,0), (0,2), (3,0), (3,2), (0,4), (3,4) \right\} \subset \BZ_6^2\ .
\end{align}
The Construction A lattice is
\begin{align}
    \Lambda(\CC) 
        &:= 
            \bigg\{
                \left( 
                \frac{ 3\,l_1}{\sqrt{6}}\ 
                ,
                \frac{ 2\,l_2}{\sqrt{6}}
                \right)
                \;\bigg|\;
                l_1,l_2 \in\BZ \,
            \bigg\} \ .
\end{align}
We see that this is just a primitive rectangular lattice with lengths determined by $C$ and $C^\perp$.
This pattern will appear
for codes over higher $k$ such as $\BZ_{12}$ as well.

For a ring $\BZ_{m_1m_2}$ with two positive integers $m_1$ and $m_2$, we can easily construct a classical code of length one as a subgroup of the cyclic group $\BZ_{m_1m_2}$.
Let $C\subset \BZ_{m_1m_2}$ be the subgroup $\BZ_{m_2}$ of $\BZ_{m_1m_2}$\,:
\begin{align}
    C = \{ \,a\,m_1\in \BZ_{m_1m_2}\mid a\in\BZ_{m_2}\,\}\,.
\end{align}
Then, the dual code is the subgroup $\BZ_{m_1}$ of $\BZ_{m_1m_2}$
\begin{align}
    C^\perp = \{\, b\,m_2\in\BZ_{m_1m_2}\mid b\in\BZ_{m_1}\,\}\,.
\end{align}
Via the CSS construction, we obtain the associated classical code $\CC = C\times C^\perp$
\begin{align}
    \CC &= \left\{\, (a\,m_1,b\,m_2) \,\big|\, a\in \BZ_{m_2},\, b\in \BZ_{m_1}\, \right\} \subset \BZ_{m_1m_2}^2\ .
\end{align}
The Construction A lattice becomes 
\begin{align}
    \Lambda(\CC) 
        = 
            \bigg\{
                \left( 
                \frac{ m_1\,l_1}{\sqrt{m_1m_2}}\ 
                ,
                \frac{ m_2\,l_2}{\sqrt{m_1m_2}}
                \right)
                \,\bigg|\,
                l_1,\,l_2 \in\BZ \,
            \bigg\} \ .
            \label{eq:n1lattice}
\end{align}

Returning to the momentum frame with (\ref{eq:frame-relation}), we see that the  left- and right-moving momenta are given by
\begin{align}\label{Ring_n=1_momenta}
    p_L &= \sqrt{\frac{m_1}{2\,m_2}}\, l_1' + \sqrt{\frac{m_2}{2\,m_1}} \,l_2' \ ,\\ 
    p_R &= \sqrt{\frac{m_1}{2\,m_2}}\, l_1' - \sqrt{\frac{m_2}{2\,m_1}}\, l_2' \ .
\end{align}
For $n=1$, \eqref{Narain_momentum_vectors} simplifies to 
\begin{align}
    p_L 
        &= 
        \frac{l}{\sqrt{G_{11}}\,r} + \frac{\sqrt{G_{11}}\,r}{2}\,w \ , \\
    p_R 
        &= 
        \frac{l}{\sqrt{G_{11}}\,r} - \frac{\sqrt{G_{11}}\,r}{2}\,w \ ,
\end{align}
where $l, w\in \BZ$. 
Comparing with \eqref{Ring_n=1_momenta}, we find the metric to be
\begin{align}
    G_{11} = \frac{1}{r^2}\,\frac{2\, m_2}{m_1} \ .
\end{align}
We thus find the radius of the compact boson $X^{(1)}$ defined by \eqref{radius_dimensionless_boson}:
\begin{align}
\label{eq:radius_n=1}
    R^{(1)} = \sqrt{\frac{2\, m_2}{m_1}} \ .
\end{align}
Taking $m_1$ and $m_2$ to be coprime positive integers, we obtain all rational CFTs with $c=1$ \cite{DiFrancesco:1997nk}.

Additionally, from T-duality, under the exchange of $l_1'$ and $l_2'$ and with $R^{(1)} \rightarrow 2/R^{(1)}$, we find the dual radius to be: 
\begin{align}
    R^{(1)} \rightarrow R'^{(1)} = \frac{2}{R^{(1)}} =  \sqrt{\frac{2\, m_1}{m_2}} \ .
\end{align}
Interestingly we have found, by considering $n=1$ non-self dual codes as a starting point, non-chiral $c=1$ rational CFT constructively.
As expected, if we consider self-dual codes ($m_1 = m_2$), we find the radius to be the self-dual radius $R^{(1)} = \sqrt{2}$. Additionally, this matches with the result that all codes of length one over $\BZ_{k^2}$ give the same CFT.

\subsubsection{Codes over $\BZ_{m_1m_2}$ with $n=2$}
For codes of length $n\ge 2$, there is a possibility of non-trivial $B$ field and background metric, but one can show that $B$ can be chosen to vanish for our Narain code CFTs as follows.

In what follows, we extend our construction for $\BZ_{m_1 m_2}$ codes with $n=1$ to the $n=2$ case.
With a variety of codes, we can consider resultant CFTs with several different metrics.

\paragraph{Case 1.}

Consider the code generated by $(m_1, \, m_2)$\,:
\begin{align}
    \label{eq:n2code}
    C = \left\{\,(a\, m_1, \,a\, m_2)\;\middle|\; a\in \BZ_{m_1m_2}\,\right\}\,,
\end{align} where $m_1$ and $m_2$ are coprime.
Then, the classical code $\CC$ is
\begin{align}
    \CC &= \left\{\, (a\, m_1, \,a\, m_2,\, a'\, m_2, \, b'\, m_1) \;\big|\; a,a',b'\in \BZ_{m_1 m_2}\,\right\} \subset \BZ_{m_1m_2}^{4}\ .
\end{align}
The Construction A lattice becomes 
\begin{align}
    \Lambda(\CC) 
        &= 
            \bigg\{ 
                \frac{ (m_1\,n_1, \, m_2\,n_2, \, m_2\,n'_1, \, m_1\,n'_2)}{\sqrt{m_1m_2}}
                \;\bigg| \;
                n_1,\,n_2,\,n_1',\,n_2' \in\BZ \,
            \bigg\}\ .
\end{align}
Clearly, the generator matrix of this lattice takes the form \eqref{eq:gen_mat_CC} with
\begin{align}
    G_{\Lambda(C)} 
        =
        \begin{bmatrix}
            ~\sqrt{\frac{m_1}{m_2}} & 0 \\
            0 & \sqrt{\frac{m_2}{m_1}} ~
        \end{bmatrix} \ ,
\end{align}
as expected, and the metric can be read off from \eqref{G_B_from_lattice} as 
\begin{align}
    G_{\mu\nu} 
        = 
        \frac{2}{r^2}\, 
            \begin{bmatrix} 
                \frac{m_2}{m_1} & 0 \\
                0 & \frac{m_1}{m_2}
            \end{bmatrix}\ .
\end{align}
Equivalently, using \eqref{eq:gen_tau_rho}, we get
\begin{align}
    \bm{\tau} = \i \,\frac{m_1}{m_2}\,,\qquad \bm{\rho} = \i\,.
\end{align}
Of course, they are in the same imaginary quadratic number field $\BQ(\i)$. Thus, the Narain code CFT is rational.

Since the metric is diagonal, we identify the CFT to be two compact $c=1$ bosons with radii dual to each other:
\begin{align}
     R^{(1)} = \sqrt{\frac{2\,m_2}{m_1}} \, , &\qquad R^{(2)} = \sqrt{\frac{2\,m_1}{m_2}} \, .
\end{align}
Note that considering the case $m_1=m_2$, or in prior words, codes over $k^2$, we obtain two copies of $c=1$ bosons at the self-dual radii $R=\sqrt{2}$. It is clear to see that in this case, we obtain an enhanced symmetry of $\SU(2) \times \SU(2)$ \cite{Dymarsky:2022kwb}.

\paragraph{Case 2.}
Consider the code generated by $(m_1, \, m_1)$.
Then, via CSS construction, we obtain the classical code
\begin{align}
    \CC &= \left\{\, (a\, m_1, \,a\, m_1,\, a', \, b'\, m_2-a') \;\big|\; a,a',b'\in \BZ_{m_1 m_2}\,\right\} \subset \BZ_{m_1m_2}^{4}\ .
\end{align}
The corresponding Construction A lattice is
 \begin{align}
    \Lambda(\CC) 
        &= 
            \bigg\{
                \frac{ (m_1\, n_1 \, ,\, m_1\, (n_1 + m_2 n_2), \, n_1' \, , \, n_2'\, m_2 - n_1')}{\sqrt{m_1m_2}}
                \;\bigg| \;
                n_1,\,n_2,\,n_1',\,n_2' \in\BZ \, 
            \bigg\} \ . \label{Z_m1m2_n=2}
\end{align}

By considering $C\subset \BZ_{m_1m_2}^2$ generated by $(m_1, \, m_1)$, we find the generator matrix of the lattice $\Lambda(C)$
\begin{equation}
    G_{\Lambda(C)} 
        = 
        \frac{1}{\sqrt{m_1m_2}} 
            \begin{bmatrix} 
                ~-m_1m_2 & 0 \\ 
                m_1 & m_1~
            \end{bmatrix}
\end{equation}
which leads to
\begin{equation}
    G_{\mu\nu} 
        = \frac{2}{r^2} 
        \begin{bmatrix} 
            \frac{2}{m_1m_2} & \frac{1}{m_1} \\ 
            \frac{1}{m_1} & \frac{m_2}{m_1} 
        \end{bmatrix} \ .
\end{equation}
In this case, we obtain
\begin{align}
    \bm{\tau} = \frac{m_2}{2} + \i\,\frac{m_2}{2}\,,\qquad \bm{\rho} = \i\,\frac{1}{m_1}\,,
\end{align}
which concludes $\bm{\tau},\bm{\rho} \in \BQ(\i)$.

\paragraph{Case 3.}
Finally, we show from considering the same construction for the code generated by $(1, \, x)$ where $x \in \BZ_{k=m_1 m_2}$ that we obtain the classical code $\CC$ 
\begin{align}
    \CC &= \left\{\, (a\, , \,a\, x, \, -b \, x, \, b) \;\big|\; a,b\in \BZ_{k}\,\right\} \subset \BZ_{k}^{4} \,,
\end{align}
which gives us the lattice
\begin{align}
    \Lambda(\CC) 
        &= 
            \bigg\{
                \frac{ (n_1 \, ,\, n_1\, x - k \,n_2,\, -n_2'\, x + k\, n_1'\, , \, n_2')}{\sqrt{k}}
                \;\bigg| \;
                n_1,\,n_2,\,n_1',\,n_2'\in\BZ \, 
            \bigg\} \, .
\end{align}
We find the generator matrix of $\Lambda(C)$ to be 
\begin{align}
    G_{\Lambda(C)} 
        = 
        \frac{1}{\sqrt{k}} 
        \begin{bmatrix} 
            ~-k~ & 0~ \\ 
            ~x & 1~
        \end{bmatrix} \ ,
\end{align}
and the metric
\begin{align}
    G_{\mu\nu} 
        = 
        \frac{2}{r^2} 
        \begin{bmatrix} 
            ~\frac{1+x^2}{k} & x~ \\ 
            ~x & k ~
        \end{bmatrix} \ .
\end{align}
In this case, the corresponding moduli are
\begin{align}
    \bm{\tau} = \frac{x\,k}{1+x^2} + \i\,\frac{k}{1+x^2}\,,\qquad \bm{\rho} = \i\,.
\end{align}
Hence, $\bm{\tau},\bm{\rho} \in\BQ(\i)$ and the Narain code CFT is rational.

\section{Narain CFTs from quantum stabilizer codes over $\BF_{p^m}$}\label{ss:finite_field}

In this section, we turn to quantum stabilizer codes over finite fields.
We start with a brief review of the fundamentals of finite fields.
Focusing on the CSS-type quantum codes, we give the construction of Narain code CFTs from stabilizer codes over $\BF_{p^m}$.
We find that all Narain code CFTs we construct are rational.
We illustrate our construction with a few examples and determine their metrics on the target space.

\subsection{Review of finite fields}
\label{ss:review_finite}

In this section, we focus on the finite field $\BF_q$ of order $q$. 
It is well known that a field of finite order exists and is unique (up to isomorphism) only when the order is a prime power $q=p^m$, where $p$ is a prime number and $m\in\BN$. 
If $q=p$ is prime, the finite field $\BF_p$ is simply the ring of integers modulo $p$, $\BZ_p=\{0,1,\dots,p-1\}$.

The standard notation for $\BF_4$ is $\BF_4=\{0,1,\omega,\bar{\omega}\}$ with $y+y=0$ for any $y\in\BF_4$, $1+\omega=\bar{\omega}$, $\omega\,\omega=\bar{\omega}$ and $\omega\,\bar{\omega}=1$. Note that $\BF_4$ is distinct from $\BZ_4$ and $\BZ_2\times\BZ_2$ as a ring.

More generally, with a basis $\{e_1,\dots,e_m\}$ of $\BF_{p^m}$ as a vector space over $\BF_p$, an element in $\BF_{p^m}$ can be written as $\sum_{i=1}^m \alpha_i\, e_i,\ \alpha_i\in\BF_p$.
Multiplication such as $e_i\,e_j$ should be fixed consistently so that division can be defined.

One canonical method to construct a finite field is by the polynomial ring. Namely, $\BF_{p^m}$ can be constructed by
\begin{equation}
    \BF_{p^m} = \BF_p[x]/(f(x))\ ,
\end{equation}
where $\BF_p[x]$ is the polynomial ring over $\BF_p$, $f(x)$ is an irreducible polynomial of degree $m$, and $(f(x))$ is the ideal generated by $f(x)$. Although there are several choices for $f(x)$ (except for $\BF_4$), all of them lead to isomorphic results.

For example, when we define $\BF_9$ by $f(x)=x^2+2x+2$, elements are expressed as $\BF_9=\{\alpha_1 x+\alpha_0 \mid \alpha_1,\alpha_0\in\BF_3 \}$ and multiplication is performed as $(x+1)(2x+1)=2x+2f(x)=2x$.

A map $\mathrm{tr}: \BF_{p^m} \to \BF_p$ is called a \emph{trace function} if it satisfies
\begin{align}
    \mathrm{tr}(a+b) &= \mathrm{tr}(a) + \mathrm{tr}(b) \,,
    \\
    \mathrm{tr}(\alpha\, a) &= \alpha \, \mathrm{tr}(a) \,,
\end{align}
for any $a,b\in\BF_{p^m}$ and any $\alpha\in\BF_p$.
For any trace functions $\mathrm{tr}_1, \mathrm{tr}_2$, there exists $k\in\BF_{p^m}$ s.t. $\mathrm{tr}_1(a)=\mathrm{tr}_2(k\,a)$ for all $a\in\BF_{p^m}$.
In particular, the standard trace function
\begin{equation}
    \mathrm{Tr}(a) = \sum_{i=0}^{m-1} a^{p^i}
\end{equation}
satisfies both conditions.

\subsection{Construction of lattices}
As shown in section \ref{ss:review_stabilizer}, we can build a quantum stabilizer code from a classical code $\CC=C_X^\perp \times C_Z^\perp$ s.t. $C_Z^\perp \subset C_X$ where the duality of codes $C_{X/Z}$ is defined with respect to the Euclidean metric. 
In this section, we construct lattices from classical codes over $\BF_{p^m}$ to associate stabilizer codes with Narain CFTs.

We first list our notations in this section. 
Let $\{e_1,...,e_m\}$ be a basis of $\BF_{p^m}$ over $\BF_p$. 
Elements $c\in\BF_{p^m}^N$ and $l\in\BR^{Nm}$ are denoted by
\begin{align}
    c=(c_1,\dots,c_N)\in\BF_{p^m}^N \,,\qquad c_i=\sum_{t=1}^m c_{i,t}\, e_t\in\BF_{p^m} \ ,
\end{align}
for $c_{i,t}\in\BF_p$ and 
\begin{align}
    l=(l_{1,1},\, \dots,\, l_{1,m},\, l_{2,1},\, \dots,\, l_{N,m})\in\BR^{Nm} \ ,
\end{align}
for $l_{i,t}\in\BR$, respectively. 
With the inclusion map $\iota: \BF_p=\{0,1,\dots,p-1\} \to \BZ$, we define a map $r: \BF_{p^m}^N \rightarrow \BR^{Nm}$ as
\begin{equation}
    r: c \mapsto \left(\, \iota(c_{1,1}),\, \dots,\, \iota(c_{1,m}),\iota(c_{2,1}),\, \dots, \, \iota(c_{N,m}) \, \right) \,.
\end{equation}

Let $C\subset \BF_{p^m}^N$ be an additive code. 
The dual code with respect to a bilinear form $\beta:\BF_{p^m}^N\times \BF_{p^m}^N\to \BF_p$ (or $\BF_{p^m}^N\times \BF_{p^m}^N\to \BF_{p^m}$) is defined by
\begin{equation}
    C^{\perp} = \left\{ c'\in \BF_{p^m}^N \;\middle|\;   \beta(c,c')=0, ~ c\in C \right\} \,.
\end{equation}

We define a lattice corresponding to the code by
\begin{equation} \label{eq:CtoL_fpm}
    \Lambda(C) := \left\{ \frac{1}{\sqrt{p}}\,v \in\BR^{Nm} \;\middle|\; v=r(c) \bmod p\,,\ c\in C \right\} \,.
\end{equation}
Note that the map $r$ and the lattice $\Lambda(C)$ depend on the choice of the basis of $\BF_{p^m}$.

The self-duality and evenness of lattices and codes are closely related to each other by the following proposition.

\begin{proposition}
Let $b:\BR^{Nm}\times\BR^{Nm}\to\BR$ and $ \beta:\BF_{p^m}^N\times \BF_{p^m}^N\to \BF_p$ be symmetric bilinear forms and suppose $b$ can be expressed as
\begin{equation}
    b(\lambda,\lambda') = \lambda\, g\,\lambda'^T \ ,
\end{equation}
for a symmetric and unimodular $g\in \mathrm{GL}(Nm,\BR)$.
We write $w_{i,t,j,s}:=g_{(i-1)m+t,(j-1)m+s}$ for simplicity.
Suppose the following condition is satisfied:
\begin{equation} \label{eq:metrics_condition}
    w_{i,t,j,s} = \iota\Bigl(\beta\bigl((0^{i-1},\, e_t,\,0^{N-i}),(0^{j-1},\,e_s,\,0^{N-j})\bigr)\Bigr) \mod p \ ,
\end{equation}
where $0^l$ is the $l$-dimensional vector whose entries are all zero.
Then, the following statements hold:
\begin{itemize}
    \item {\bf (Self-duality)}
    $\Lambda(C)$ is self-dual with respect to $b$ if and only if $C$ is self-dual with respect to $\beta$.
    \item {\bf (Evenness)}
    When $p=2$, $\Lambda(C)$ is even with respect to $b$ if and only if 
    \begin{equation} \label{eq:evenness_p2}
        \sum_{i,t,j,s} \iota(c_{i,t})\,\iota(c_{j,s})\,w_{i,t,j,s} \in 4\BZ
    \end{equation}
    for any $c\in C$.
    When $p$ is odd prime, $\Lambda(C)$ is even with respect to $b$ if and only if 
    \begin{equation} \label{eq:evenness_podd}
        \beta(c,c)=0 
    \end{equation}
    for any $c\in C$ and all diagonal elements of $g$ are even, i.e., $w_{i,t,i,t} \in 2\,\mathbb{Z}$ for any $i,t$.
\end{itemize}
\label{prop:even_sf_finite}
\end{proposition}

\bigskip

\noindent
The proof is slightly technical and deferred to Appendix~\ref{app:proof}.

\begin{corollary}
Suppose a basis $\{e_1,...,e_m\}$ of $\BF_{p^m}$ and a trace function $\mathrm{tr}: \BF_{p^m}\to\BF_p$ satisfy
\begin{equation} \label{eq:field_orthonormality}
    \delta_{t,s} = \mathrm{tr}(e_t\,e_s) \,,
\end{equation}
for any $t,s\in\{1,\dots,m\}$.
Let $\xi: \BF_{p^m}^N\times \BF_{p^m}^N\to \BF_{p^m}$ be a bilinear form that can be written as
\begin{equation}
    \xi(c,c')=c\, \tilde{h}\, c'^T\,,\qquad \tilde{h}\in\mathrm{GL}(N,\BF_p)\,,
\end{equation}
where $\tilde h$ satisfies the condition $h=\iota(\tilde{h}) \mod p$ for some symmetric and unimodular matrix $h\in\mathrm{GL}(N,\BR)$.
Then a linear code $C\subset \BF_{p^m}^N$ is self-dual with respect to $\xi$ if and only if $\Lambda(C)$ is self-dual with respect to the metric $h\otimes I_m$.
\end{corollary}

\begin{proof}
In the previous proposition, if we set
\begin{equation}
    b(\lambda,\lambda') = \lambda\, (h\otimes I_m)\, \lambda'^T \,,\qquad
    \beta(c,c') = \mathrm{tr}\left(\xi(c,c')\right)\,,\
\end{equation}
then the condition \eqref{eq:metrics_condition} becomes \eqref{eq:field_orthonormality}. In addition, the self-duality with $\xi(c,c')$ implies the self-duality with $\mathrm{tr}(\xi(c,c'))$.
\end{proof}

To construct non-chiral CFTs, we set $N=2n$ and take the metric to be
\begin{equation}
    h = \begin{bmatrix} 
        ~0~ & I_n~ \\
        ~I_n~ & 0~ 
        \end{bmatrix}\,.
\end{equation}
Then, the bilinear form $\xi$ reduces to the off-diagonal Lorentzian metric.
Let us consider a linear code $C\subset \BF_{p^m}^n$ and define a dual code $C^\perp \subset \BF_{p^m}^n$ with respect to the Euclidean metric as \eqref{eq:dual_code_Euclidean}. Then $\CC = C\times C^\perp \subset \BF_{p^m}^{2n}$ is self-dual with respect to the off-diagonal Lorentzian metric since codewords $c=(a,b), c'=(a',b')\in \CC \, (a,a'\in C,\, b,b'\in C^\perp)$ satisfy
\begin{equation}
    \xi(c,c') = a\cdot b' + a'\cdot b = 0
\end{equation}
and $|\CC|=|C||C^\perp|=p^{mn}$. Therefore, if we can take a basis and a trace function for $\BF_{p^m}$ that satisfy the orthonormality \eqref{eq:field_orthonormality}, the CSS-type stabilizer codes constructed from $\CC = C \times C^\perp$ over $\BF_{p^m}$ always generate self-dual lattices.

Regarding the evenness, for $p=2$, \eqref{eq:evenness_p2} is satisfied since
\begin{equation}
    2 \sum_{i,t} \iota(c_{i,t})\, \iota(c_{i+n,t}) = 2 \sum_{i,t} \iota(a_{i,t})\, \iota(b_{i,t}) \equiv 2\, \iota(\mathrm{tr}(a\cdot b)) \equiv 0 \mod 4 \,.
\end{equation}
For odd prime $p$, all diagonal elements of the off-diagonal Lorentzian metric are even (0) and \eqref{eq:evenness_podd} is satisfied when the code is self-dual. Thus, for any prime $p$, self-dual lattices constructed from the CSS codes are also even, which can yield Narain CFTs.

\subsection{Narain code CFTs}
Once an even self-dual lattice with respect to the off-diagonal Lorentzian metric is obtained, we can construct the corresponding CFT as in the ring case. In other words, the left- and right-moving momenta of the states $\ket{p_L,p_R}$ in the CFT are defined by
\begin{equation}
    \left( \frac{p_L + p_R}{\sqrt{2}}, \frac{p_L - p_R}{\sqrt{2}}\right) \in \Lambda(\CC)\ ,
\end{equation}
and the Hilbert space is built based on \eqref{eq:Hilbert_space}.

The torus partition function of the CFT corresponding to a code $\CC=C\times C^\perp \subset \BF_{p^m}^{2n}$ is
\begin{equation}
    Z_\CC(\tau,\bar{\tau}) = \frac{1}{|\eta(\tau)|^{2nm}} \sum_{(\Ba,\Bb)\,\in\, \CC} \;\sum_{l_1,l_2 \,\in\, \BZ^{nm}} q^{\frac{p}{4} \left(\frac{r(\Ba)+r(\Bb)}{p} + l_1+l_2\right)^2} \bar{q}^{\frac{p}{4} \left(\frac{r(\Ba)-r(\Bb)}{p} + l_1-l_2\right)^2} \, ,
\end{equation}
where $q = e^{2\pi\i\tau}$ with the torus moduli $\tau$.
As in the ring case, we can express the partition functions in terms of the complete enumerator polynomial
\begin{equation}
    W_\CC(\{x_{ab}\}) = \sum_{c \,\in\, \CC} \;\prod_{(a,b)\,\in\, \BF_{p^m} \times \BF_{p^m}} x_{ab}^{\text{wt}_{ab}(c)}\ ,
\end{equation}
where $\text{wt}_{ab}(c)$ is the number of components $(\Ba_i, \Bb_i)\in \BF_{p^m} \times \BF_{p^m}$ equal to $(a, b)$ for a codeword $c\in\CC$\,:
\begin{equation}
    \text{wt}_{ab}(c) = |\,\{i \mid (\Ba_i, \Bb_i)=(a,b),\, c=(\Ba_1,\dots,\Ba_n,\Bb_1,\dots,\Bb_n)\}\,| \,.
\end{equation}
In fact, substituting
\begin{equation}\label{eq:Psi-ab}
\begin{aligned}
    \Psi_{ab} &= \sum_{l_1,l_2\, \in \,\BZ^{m}} q^{\frac{p}{4} \left(\frac{r(a)+r(b)}{p} + l_1+l_2\right)^2} \bar{q}^{\frac{p}{4} \left(\frac{r(a)-r(b)}{p} + l_1-l_2\right)^2} \\
    &= \prod_{s=1}^m \psi_{r(a)_s r(b)_s}\ ,
\end{aligned}
\end{equation}
where $\psi$ is defined in \eqref{eq:psi_ab} ($k=p$), we get
\begin{equation}\label{eq:ZC-WC-Psi}
     Z_\CC(\tau,\bar{\tau}) = \frac{1}{|\eta(\tau)|^{2nm}}\,  W_\CC(\{\Psi_{ab}\}) \,.
\end{equation}
We note that $\Psi_{ab}$ in~(\ref{eq:Psi-ab}) can be decomposed into a finite sum involving~$\Theta_{m,p}(\tau)$ defined in~(\ref{eq:Theta-m-k}) and its complex conjugate, as in~(\ref{eq:psi-ab-Theta}).
Thus the partition function is a finite sum of products of $\U(1)_{2p}$ characters, and the CFT is rational.

\subsection{Examples}

We first consider whether we can take an orthonormal basis satisfying \eqref{eq:field_orthonormality}.

For $\BF_4$, we use the standard notation $\BF_4=\{0,1,\omega,\bar{\omega}\}$. If we take $e_1=\omega,\, e_2=\bar{\omega}$ as the basis and use the standard trace $\Tr(x)=x+x^2$, the orthonormality is satisfied since $\Tr(0)=\Tr(1)=0,\, \Tr(\omega)=\Tr(\bar{\omega})=1$.

For general finite fields $\BF_{p^m}$, it is convenient to use the expression by the polynomial ring over $\BF_p$, i.e., $\BF_{p^m} = \BF_p[x] / (f(x))$. As concrete examples, we consider the case $m=2,3$.

For $m=2$, we take $e_1 = 1,\, e_2=x$ as the basis and use the trace $\mathrm{tr}(a_1\,x+a_0) = a_0$.
An orthonormal basis is given by setting
\begin{align}
     f(x) = x^2 - f_1\, x - 1\ , \qquad f_1\in\BF_p\, , \label{eq:irr_poly_m2}
\end{align}
such that $f(x)$ is irreducible.
For example, $f(x)$ can be chosen as
\begin{align}
    \BF_4&: ~x^2+x+1\ ,\\
    \BF_9&: ~x^2+2x+2\ ,\\
    \BF_{25}&: ~x^2+2x+4 \ .
\end{align}
For $m=3$, we take $e_1 = 1,\, e_2=x,\, e_3=x^2-1$ as the basis and use the trace $\mathrm{tr}(a_2\,x^2+a_1\,x+a_0) = a_2 + a_0$.
An orthonormal basis is given by setting
\begin{align}
    f(x) = x^3 - f_0\, x^2 - 2x + f_0 \ , \qquad f_0\in\BF_p \,,
    \label{eq:irr_poly_m3}
\end{align}
such that $f(x)$ is irreducible.
For example, $f(x)$ can be chosen as
\begin{align}
    \BF_8&: ~x^3+x^2+1\ ,\\
    \BF_{27}&: ~x^3+x^2+x+2\ ,\\
    \BF_{125}&: ~x^3+x^2+3x+4 \ .
\end{align}

For any $\BF_{p^2},\BF_{p^3}$, there exist $f_1, f_0\in\BF_p$ in \eqref{eq:irr_poly_m2}, \eqref{eq:irr_poly_m3} that make $f(x)$ be irreducible.
To show this statement, we use the following proposition:

\begin{proposition}
    Let $f(x;a)\in\BF_p[x]$ be a polynomial of degree 2 or 3 with a parameter $a\in\BF_p$. If the conditions
    \begin{itemize} 
        \item[(I)] there exists $b\in\BF_p$ s.t. $f(b;a)\neq0$ for any $a\in\BF_p$,
        \item[(II)] for any $b\in\BF_p$, there exists at most one $a\in\BF_p$ s.t. $f(b;a)=0$,
    \end{itemize}
    hold, then $f(x;a)\in\BF_p[x]$ is irreducible for some $a\in\BF_p$.
\end{proposition}

\begin{proof}
    We assume that for any $a\in\BF_p$ the polynomial $f(x;a)\in\BF_p[x]$ is reducible and derives a contradiction. When the degree is 2 or 3, the factorization of $f(x;a)$ must include a degree-one polynomial such as $x-b,\,b\in\BF_p$. Therefore, the assumption is equivalent to $\forall a\in\BF_p,\, \exists b\in\BF_p,\, f(b;a)=0$. If the condition (I) is satisfied, since both $a$ and $b$ can take $p$ different values, there exists at least one $b\in\BF_p$ s.t. $f(b;a_1)=f(b;a_2)=0$ for $a_1\neq a_2\in\BF_p$. However, this contradicts condition (II).
\end{proof}
Using the proposition, we can check that the polynomials \eqref{eq:irr_poly_m2} and \eqref{eq:irr_poly_m3} become irreducible for some $f_1, f_0\in\BF_p$ as follows.

For $f(x;f_1) = x^2 - f_1\, x - 1 \in \BF_p[x]$ with $f_1\in\BF_p$ (\eqref{eq:irr_poly_m2}), $b=0$ satisfies the condition (I) since $f(0;f_1)=-1$. 
The condition (II) is also satisfied since if $f(b;a_1)=f(b;a_2)=0$ for $a_1\neq a_2\in\BF_p$, then $0=f(b;a_1)-f(b;a_2)=(a_2-a_1)\,b$, which can be satisfied only if $b=0$. 
However, $f(0;f_1)=-1$ as we have already seen.

For $f(x;f_0) = x^3 - f_0\, x^2 - 2x + f_0 \in \BF_p[x]$ with $f_0\in\BF_p$ (\eqref{eq:irr_poly_m3}), $b=1$ satisfies the condition (I) since $f(1;f_0)=-1$. 
The condition (II) is also satisfied since if $f(b;a_1)=f(b;a_2)=0$ for $a_1\neq a_2\in\BF_p$, then $0=f(b;a_1)-f(b;a_2)=(a_2-a_1)(b^2-1)$, which can be satisfied only if $b^2=1$. 
However, in this case $0=f(b;a_1)=(b-a_1)(b^2-1)-b=-b$, which is inconsistent.

We are now in a position to examine a few examples using the above expressions.
As in the $\BZ_k$ case, we have $B_{\mu\nu}=0$ as follows from~\eqref{G_B_from_lattice}.

\paragraph{$\BF_4$ with $n=1$ ($c=2$).}

We consider a code $C\subset\BF_4^1$\,:
\begin{equation}
    C = \{0,1,\omega,\bar{\omega}\} \,.
\end{equation}
The dual code is
\begin{equation}
    C^\perp = \{0\}  \,.
\end{equation}
Since the coefficients of the basis $e_1=\omega, e_2=\bar{\omega}$ are
\begin{equation}
    r(0)=(0,0) \,,\quad r(1)=(1,1) \,,\quad r(\omega)=(1,0) \,,\quad r(\bar{\omega})=(0,1) \,,
\end{equation}
the lattice $\Lambda(\CC)$ for the code $\CC=C\times C^\perp \subset \BF_4^2$ is
\begin{equation}
    \Lambda(\CC) = \left\{ \frac{(n_1,\, n_2,\, 2n_1',\, 2n_2')}{\sqrt{2}} \;\middle|\; n_1,\, n_2,\, n_1',\, n_2'\in\BZ \, \right\} \,.
\end{equation}

From the same discussion as in the ring case, we can fix the metric and anti-symmetric tensor for the corresponding Narain code CFT using \eqref{G_B_from_lattice}:
\begin{equation}
    G_{\mu\nu} = \frac{4}{r^2}\, \delta_{\mu\nu}  \,.
\end{equation}
In this case, the complex parameters \eqref{eq:gen_tau_rho} are
\begin{align}
    \bm{\tau} = \i \ , \qquad \bm{\rho} = 2\,\i \ ,
\end{align}
which are in the same imaginary quadratic field, hence the theory is rational.

The complete enumerator polynomial of the code is
\begin{equation}
    W_\CC(\{x_{ab}\}) = x_{00} + x_{10} + x_{\omega0} + x_{\bar{\omega}0} \ ,
\end{equation}
and the partition function of the CFT can be expressed as
\begin{align}
\begin{aligned}
    Z_\CC(\tau,\bar{\tau}) 
        &= 
        \left[ \frac{1}{|\eta(\tau)|^2} \frac{\theta_2\bar{\theta}_2 + \theta_3\bar{\theta}_3 + \theta_4\bar{\theta}_4}{2} \right]^2 \ ,
\end{aligned}
\end{align}
where we used \eqref{eq:psi-ab-theta}. This is the partition function for a pair of Dirac fermions, each of which is gauged by the fermion parity (GSO projection).

\paragraph{$\BF_4$ with $n=2$ ($c=4$).}
We consider a self-dual code $C\subset\BF_4^2$\,:
\begin{equation}
    C = \{(0,0),(1,1),(\omega,\omega),(\bar{\omega},\bar{\omega})\} = C^\perp \,.
\end{equation}
The lattice $\Lambda(\CC)$ for the code $\CC=C\times C \subset \BF_4^4$ is
\begin{equation}
    \Lambda(\CC) = \left\{ \frac{(n_1,\, n_2,\, n_1+2n_3,\, n_2+2n_4,\, n_1',\, n_2',\, n_1'+2n_3',\, n_2'+2n_4')}{\sqrt{2}} \;\middle|\; n_i,\, n_i' \in\BZ \, \right\} \,.
\end{equation}
The metric of the corresponding CFT is
\begin{equation}
    G_{\mu\nu} = \frac{2}{r^2} \begin{bmatrix}
        2&0&-1&0 \\
        0&2&0&-1 \\
        -1&0&1&0 \\
        0&-1&0&1
    \end{bmatrix} 
    \,.
\end{equation}

Using the complete enumerator polynomial of the code 
\begin{align}
    W_\CC(\{x_{ab}\}) 
    = \sum_{a,b\in\BF_4} x_{ab}^2 \,, 
\end{align}
the partition function of the CFT can be fixed as
\begin{align}
    \begin{aligned}
        Z_\CC(\tau,\bar{\tau}) 
            &= 
            \left[ \frac{1}{|\eta(\tau)|^4} \,\frac{\theta_2^2\,\bar{\theta}_2^2 + \theta_3^2\,\bar{\theta}_3^2 + \theta_4^2\,\bar{\theta}_4^2}{2} \right]^2 \,.
    \end{aligned}
\end{align}
This is the partition function for two copies of a pair of Dirac fermions, where each copy is gauged by the diagonal $\BZ_2$ symmetry of the fermion parity (GSO projection).

\paragraph{$\BF_9$ with $n=2$ ($c=4$).}
We consider a self-dual code $C\subset\BF_9^2$ generated by $(1,\, x+1)$, i.e., the code and dual code are
\begin{align}
    \begin{aligned}
        C = C^\perp
            =
            \big\{ &(0,0),(1,x+1),(2,2x+2),(x,2x+1),\\
                & (x+1,2), (x+2,x),(2x,x+2),(2x+1,2x),(2x+2,1)\big\}\,.
    \end{aligned}
\end{align}
For example,
\begin{equation}
    r\left((1,x+1)\right) = (1,0,1,1) \,,\quad r\left((x,2x+1)\right) = (0,1,1,2) \,,
\end{equation}
thus the Construction A lattice $\Lambda(\CC)$ for the code $\CC=C\times C \subset \BF_9^4$ consists of the elements
\begin{equation}
\begin{aligned}
    \frac{1}{\sqrt{3}} (n_1,\, n_2,\, n_1+n_2+3n_3,\, n_1+2n_2+3n_4,\, n_1',\, n_2',\, n_1'+n_2'+3n_3',\, n_1'+2n_2'+3n_4') \,,
\end{aligned}
\end{equation} where $n_i,n_i'\in\BZ$.
The metric of the corresponding CFT is
\begin{equation}
    G_{\mu\nu} = \frac{2}{r^2} \begin{bmatrix}
        3&0&-1&-1 \\
        0&3&-1&-2 \\
        -1&-1&1&1 \\
        -1&-2&1&2
    \end{bmatrix} 
    \,.
\end{equation}

Although we do not write down the complete enumerator polynomial since the number of terms is large, we can express the partition function using it as in the previous examples.

\paragraph{$\BF_8$ with $n=2$ ($c=6$).}
We consider a code $C\subset\BF_8^2$ generated by $(1,\, x)$, i.e., the code is
\begin{align}
    \begin{aligned}
        C 
        = 
        \big\{ &(0,0),(1,x),(x,x^2),(x+1,x^2+x),(x^2,x^2+1), \\
            &(x^2+1,x^2+x+1),(x^2+x,1),(x^2+x+1,x+1)\big\} \,.
    \end{aligned}
\end{align}
The dual code is generated by $(1,\, x^2+x)$.
Note that the basis is $e_1=1,\, e_2=x,\, e_3=x^2+1$ (not $e_3=x^2$) and then for example
\begin{equation}
    r\left((1,x)\right) = (1,0,0,0,1,0) \,,\quad r\left((x,x^2)\right) = (0,1,0,1,0,1) \,.
\end{equation}
The metric of the corresponding CFT is
\begin{equation}
    G_{\mu\nu} = \frac{2}{r^2} \begin{bmatrix}
        2&0&0&0&-1&0 \\
        0&2&0&-1&0&-1 \\
        0&0&2&0&-1&-1 \\
        0&-1&0&1&0&\frac{1}{2} \\
        -1&0&-1&0&\frac{3}{2}&\frac{1}{2} \\
        0&-1&-1&\frac{1}{2}&\frac{1}{2}&\frac{3}{2}
    \end{bmatrix} \,.
\end{equation}

\section{Narain CFTs for CSS codes with non-zero logical qubits}
\label{sec:non-zero-qubits}

In this section, we propose a slightly extended correspondence between quantum stabilizer codes and Narain CFTs, focusing on the CSS codes in the binary~$(\mathbb{F}_2)$ case.
Namely, one associates, to a quantum stabilizer code with non-zero logical qubits, a  set of Narain CFTs.
We first explain the motivations and the general correspondence, and then give concrete examples.

As discussed in the literature and also in this paper, we have a procedure to construct a Narain CFT from a classical code $\mathcal{C}$ self-dual and doubly-even with respect to the off-diagonal Lorentzian inner product $\eta$ defined by~(\ref{eq:lor-inner}).
Although the self-duality of the classical code is natural from the perspective of Narain CFTs because it corresponds to the self-duality of the momentum lattice, it implies that the corresponding quantum code has a one-dimensional code subspace and zero logical qubits.
It is desirable to introduce non-zero logical qubits into the correspondence between quantum codes and Narain CFTs.
This can be achieved as follows.
Suppose that a classical code $\mathcal{C}_0$ is doubly-even and self-orthogonal but not self-dual.
Suppose also that there exists a self-dual doubly-even code  $\mathcal{C}$ such that  $\mathcal{C}_0\subset  \mathcal{C}$.
For a general self-orthogonal~$\mathcal{C}_0$, such~$\mathcal{C}$, even if it exists, may not be unique.
We can consider a correspondence between~$\mathcal{C}_0$ and the set of such~$\mathcal{C}$'s.
In other words, when the assumptions above are satisfied, {\it a doubly-even self-orthogonal code corresponds to a finite set of Narain code CFTs constructed from the doubly-even self-dual codes that contain the self-orthogonal one}.%
\footnote{%
We expect that this extended correspondence will lead to a unified picture when lifted to a three-dimensional set-up involving abelian Chern-Simons theories along the lines of~\cite{Kawabata:2023iss,aharony2023holographic}.  It is straightforward to generalize the correspondence to qudit CSS codes.}

Before going through examples, let us briefly review the check matrix representation of CSS codes in the binary case.
Suppose that a stabilizer group $S$ is generated by $(n-k)$ generators $g(\Bu^{(1)}),\cdots, g(\Bu^{(n-k)})$ where $\Bu^{(i)} = (\Ba^{(i)}\mid \Bb^{(i)})\in \BF_2^{2n}$.
Then, we can encode the stabilizer group $S$ into the check matrix
\begin{align}
\label{eq:check_matrix}
    \textsf{H} = \left[
    \begin{array}{c|c}
    \Ba^{(1)} & \Bb^{(1)} \\
    \Ba^{(2)} & \Bb^{(2)} \\
    \vdots & \vdots \\
    \Ba^{(n-k)} \;& \;\Bb^{(n-k)}
    \end{array}\right]\,.
\end{align}
The commutativity of the stabilizer group $S$ is represented by the orthogonality $\mathsf{H}\,\eta\, \mathsf{H}^T = 0$ mod $2$ with respect to $\eta$.
The correspondence between stabilizer generators and a check matrix establishes an alternative representation of quantum stabilizer codes by orthogonal geometry \cite{Calderbank:1996hm}.

The formulation based on a check matrix is useful to define a CSS code from a pair of classical linear codes.
For a linear code encoding $k$ bits into $n$ bits, a generator matrix $G$ is defined as a $k\times n$ matrix of rank $k$ that generates $C$\,:
\begin{align}
    G = \left\{ \,x\,G\in\BF_2^n\mid x\in \BF_2^k\right\}\,.
\end{align}
On the other hand, a parity check matrix $H$ is an $(n-k)\times n$ matrix of rank $n-k$ satisfying $G\,H^T = 0$ mod 2, which generates the dual code $C^\perp$ with respect to the Euclidean inner product
\begin{align}
    C^\perp = \left\{\,x\,H\in\BF_2^n\;\left|\; x\in\BF_2^{n-k}\right.\right\}\,.
\end{align}
We call a linear code $[n,k]$ type if it encodes $k$ bits into $n$ bits.

Suppose $C_X$ and $C_Z$ are $[n,k_X]$ and $[n,k_Z]$ linear codes satisfying $C_Z^\perp \subset C_X$.
The CSS code can be defined by the check matrix
\begin{align}
\label{eq:CSS}
    \mathsf{H}_{(C_X,\,C_Z)} = 
    \left[
    \begin{array}{c|c}
        H_X\, & 0  \\
        0 &\, H_Z
    \end{array}
    \right],
\end{align}
where the block $H_X$ ($H_Z$) represents the parity check matrix of the binary linear code $C_X$ ($C_Z$).
Due to the condition $C_Z^\perp\subset C_X$, the matrix satisfies the orthogonality relation $\mathsf{H}_{(C_X,\,C_Z)}\,\eta\,\mathsf{H}_{(C_X,\,C_Z)}^T = 0$ mod $2$ and defines a stabilizer group $S$.
Conventionally, a $2^k$-dimensional subspace of $\BC^{2^n}$ with distance $d$ is called an $[[n,k,d]]$ quantum stabilizer code.
We can check that the CSS code is $[[n,k_X+k_Z-n]]$ type where we omit the distance $d$.

From the CSS code, we obtain a classical code
\begin{align}
    \CC_0 = \left\{ x\,\mathsf{H}_{(C_X,\,C_Z)}\in\BF_2^{2n}\mid x\in\BF_2^{2n-k_X-k_Z} \right\} =  C_X^\perp\times C_Z^\perp\,.
\end{align}
As the orthogonality relation is satisfied, the classical code $\CC_0 = C_X^\perp\times C_Z^\perp$ is self-orthogonal with respect to the off-diagonal Lorentzian metric $\eta$. 
If the self-orthogonal code $\CC_0$ can be extended to a finite set of doubly-even self-dual codes $\CC$, then it corresponds to a finite set of Narain code CFTs.
In what follows, we explain the correspondence by taking Shor's code and Steane's code as examples.

\subsection{$[[9,1,3]]$ Shor code}

Shor's code~\cite{PhysRevA.52.R2493} is a CSS code defined by $[9,7]$ and $[9,3]$ linear codes $C_X$ and $C_Z$, respectively, whose parity check matrices are (\cite{Shor-zoo})
\begin{equation}
H_X = 
\left[
\begin{array}{ccc|ccc|ccc}
1&1&1~&~1&1&1~&~0&0&0 \\
\hline
1&1&1~&~0&0&0~&~1&1&1
\end{array}
\right]\,,
\end{equation}
and
\begin{equation}
    H_Z 
        = 
        \left[
            \begin{array}{ccc|ccc|ccc}
                1&1&0~&~0&0&0~&~0&0&0 \\
                0&0&0~&~1&1&0~&~0&0&0 \\
                0&0&0~&~0&0&0~&~1&1&0 \\
                \hline
                1&0&1~&~0&0&0~&~0&0&0 \\
                0&0&0~&~1&0&1~&~0&0&0 \\
                0&0&0~&~0&0&0~&1&0&1
            \end{array}
        \right] \,.
\end{equation}
The logical codewords are
\begin{align}
    |\bar 0\rangle &= 2^{-3/2} (|000\rangle + |111\rangle)^{\otimes 3}  \,, \\
    |\bar 1\rangle &= 2^{-3/2} (|000\rangle - |111\rangle)^{\otimes 3}  \,.
\end{align}
Logical operators can be represented as
\begin{align}
    \bar X & = Z_1\, Z_2\, \cdots\, Z_9 \,, \\
    \bar Z & = X_1\, X_2\,\cdots\, X_9 \,. 
\end{align}
We can introduce new stabilizers of the form ${\rm i}^{ab}\, \bar{X}^a\, \bar{Z}^b$ so that the number of logical qubits becomes zero.
Let $\mathcal{S}$ be a set of new stabilizers. 
There are two
possibilities: 1) $\mathcal{S} =  \mathcal{S}_1 := \{I, \bar{X}\}$ and
2)  $\mathcal{S} = \mathcal{S}_2:= \{I, \bar{Z}\}$.
These correspond to classical codes $\mathcal{C}_i$ ($i=1,2$) generated by
\begin{equation}
\mathsf{H}_i = \left[
\begin{array}{c|c}
 ~H_X \;& 0~ \\
~0&\;H_Z~ \\
~{\bf a}_i&{\bf b}_i~
\end{array}\right]
\end{equation}
with
\begin{equation}
{\bf a}_1 = 
[1^9]
\,, \quad 
{\bf b}_1 = [0^9] \,,
\end{equation}
\begin{equation}
{\bf a}_2 = [0^9]\,, \quad 
{\bf b}_2 = [1^9]
 \,.
\end{equation}
Here the power $x^n$ indicates $n$ repeated entries of $x$.
Thus, Shor's code corresponds to the two
Narain code CFTs constructed from $\mathcal{C}_i$ ($i=1,2$).%
\footnote{%
The choice ${\bf a}_3 = [1^9]$, ${\bf b}_3 = [1^9]$ leads to a self-dual code that is not doubly-even (singly-even) and corresponds to a fermionic CFT.}
We note that the two classical codes give rise to CSS codes.

\subsection{$[[7,1,3]]$ Steane code}

Steane's code \cite{steane1996multiple} is also a CSS code defined by two copies of a $[7,4]$ Hamming code $C_X = C_Z = C$ whose parity check matrix is (\cite{Steane-zoo})
\begin{equation}
H_X=H_Z = H:=
\begin{bmatrix}
~1&0&0&1&0&1&1~ \\
~0&1&0&1&1&0&1~ \\
~0&0&1&0&1&1&1~
\end{bmatrix} \,.
\end{equation}
The logical codewords are
\begin{align}
|\bar 0\rangle &= 2^{-3/2} \Big(  |0000000\rangle +  |1010101\rangle + |0110011\rangle + |1100110\rangle \\ \nonumber
&\qquad\qquad\qquad
  + |0001111\rangle + |1011010\rangle + |0111100\rangle + |1101001\rangle \Big) \,,
\\
|\bar 1\rangle &= X_1\,X_2\,\cdots\, X_7\, |\bar 0\rangle  \,.
\end{align}
Logical operators can be represented as
\begin{align}
\bar X &=  X_1\,X_2\,\cdots\, X_7 \,, \\
\bar Z &= Z_1\,Z_2\,\cdots\, Z_7 \,.
\end{align}
Again, there are two self-dual codes $\mathcal{C}_i$ ($i=1,2$) that contain $\mathcal{C}_0 = C^\perp\times C^\perp$.
These are given by the generator matrices
\begin{equation}
\mathsf{H}_i = 
\begin{bmatrix}
~H & 0~ \\
~0 & H~ \\
~{\bf a}_i & {\bf b}_i ~
\end{bmatrix} \,,
\end{equation}
where
\begin{equation}
{\bf a}_1 = 
[1^7]
\,, \quad 
{\bf b}_1 = [0^7] \,,
\end{equation}
\begin{equation}
{\bf a}_2 = [0^7]\,, \quad 
{\bf b}_2 = [1^7]
 \,.
\end{equation}
Thus, the Steane code corresponds to the two
Narain code CFTs constructed from $\mathcal{C}_i$ ($i=1,2
$).%
\footnote{%
As in the previous example, the choice ${\bf a}_3 = [1^7]$, ${\bf b}_3 = [1^7]$ gives a fermionic CFT. See also \cite{Dixon:1988qd,Gaiotto:2018ypj,Kawabata:2023rlt,Kawabata:2023usr,Kawabata:2023iss,Kawabata:2023nlt,Moore:2023zmv} for constructions of fermionic CFTs from classical and quantum codes.}
Again, the two classical codes give rise to CSS codes.

\section{Discussion}\label{ss:discussion}

In this paper, we have considered the construction of Narain CFTs from quantum stabilizer codes and generalized the previous construction \cite{Dymarsky:2020qom,Kawabata:2022jxt} to rings of integers modulo $k$ and finite fields of order $q = p^m$.
We exploited the general relationship between quantum stabilizer codes and classical codes over finite Frobenius rings to provide Lorentzian even self-dual lattices via Construction A.
Using various examples of quantum codes, we examined our construction of Narain CFTs and identified Narain moduli parameters of the resulting theories.
Also, we proposed an extension of the conventional correspondence between quantum codes and Narain CFTs so that the new correspondence is applicable to quantum codes with non-zero logical qubits.

A recent paper~\cite{Furuta:2023xwl} studied the sufficient conditions on the metric and B-field for a Narain CFT to have the structure of a Narain code CFT associated with a code over a finite field $\mathbb{F}_p$ with $p$ prime or a ring such as $\mathbb{Z}_k$ with $k$ an integer, where the information on the code structure can be read from the character decomposition of the torus partition function.
It would be interesting to do a similar analysis allowing the Narain CFT to correspond to a code over a more general Frobenius algebra including~$\mathbb{F}_{p^m}$ with $m>1$.

This paper has focused on the torus partition function and found out that they can be represented by complete weight enumerators.
Generally, we can place a Narain code CFT on a higher-genus Riemann surface.
In the case of binary stabilizer codes, it is known that a higher-genus partition function of a Narain code CFT can be computed with the help of a higher-genus complete weight enumerator~\cite{Henriksson:2022dnu}.
The generalization to our Narain code CFTs is straightforward and may be useful for performing modular bootstrap.

In section~\ref{sec:non-zero-qubits}, we proposed an extended correspondence between stabilizer codes with a non-zero number of logical qubits and a finite set of Narain code CFTs.
As examples, we considered two famous CSS codes, the Shor code and the Steane code.
It would be interesting to see if these sets of Narain CFTs have nice physical properties.

Motivated by the works~\cite{Maloney:2020nni,Afkhami-Jeddi:2020ezh}, the authors of~\cite{Dymarsky:2020qom,Dymarsky:2020pzc,Angelinos:2022umf,Henriksson:2022dml,Kawabata:2022jxt} computed the spectral gap
for the average over an ensemble of Narain CFTs.
It would be interesting to generalize these results and perform averaging over self-dual codes over a ring such as $\mathbb{Z}_k$ or a finite field $\mathbb{F}_{p^m}$ with $m>1$, and study the holographic dual of averaged Narain code CFTs. There exist constructions of ensemble averages over Narain code CFTs over $\BF_{p^{m=1}}$ and $\BZ_{k=p}$~\cite{Dymarsky:2020bps,Kawabata:2022jxt,aharony2023holographic,Barbar:2023ncl} where the large order limit of the rings were taken and the partition functions were derived, though not for arbitrary $m>1$ and $k\neq p$, as well as the large central charge limit. 
It would be interesting to study the theories that arise when taking different large order vs large central charge limits.
Recently, \cite{aharony2023holographic} averaged over length $n$ Narain code CFTs constructed over $\BZ_k\times\BZ_k$ for prime $k$ and showed it to be dual to level-$k$ $(\U(1)\times \U(1))^n$ Chern-Simons theories on different handlebodies.
Given a way to correctly weight (and count) all the codes in a ring of set length, an ensemble average can be extended to arbitrary $k$ and $m>1$ in a similar manner to \cite{Angelinos:2022umf,Kawabata:2022jxt}.\footnote{
Self-dual codes over $\BZ_k$ are classified for $k\le 24$ and $n\le 9$ \cite{harada2009classification,harada2016classification}.
See also the database of self-dual codes \cite{munemasa_database} over various rings and finite fields for more information.
The mass formula which gives the number of distinct self-orthogonal codes over $\BZ_{p^2}$ for prime $p$ is obtained in \cite{balmaceda2008mass,betty2008mass}, which may be useful to taking ensemble averaging of the associated Narain code CFTs.
}
The authors of~\cite{aharony2023holographic} also noted a relation between toric codes and ``AB" Chern-Simons theories (dual to $c=1$ compact scalar theories), which would be an interesting direction to pursue in the context of Narain code CFTs.

Narain code CFTs have been used to construct a concrete theory with a large spectral gap~\cite{Yahagi:2022idq,Furuta:2022ykh}.
In particular, the work~\cite{Angelinos:2022umf} showed that Narain theories with the largest spectral gap with central charge $c\leq8$, which was identified in~\cite{Afkhami-Jeddi:2020ezh}, can be constructed from codes. 
In our construction, we have a larger number of CSS codes than before, and we can use classical self-dual codes to search for CFTs with a large spectral gap.
It is interesting to find CSS codes that give Narain CFTs with a large spectral gap by our construction.

From \eqref{eq:radius_n=1}, we can see that codes over rings construct a compact boson of self-dual radius.
At self-dual radius, a compact boson has been known to have the extended symmetry $\mathrm{SU}(2)\times \mathrm{SU}(2)$. 
Although Narain code CFTs manifestly have $\mathrm{U}(1)$ symmetry, it is not clear when the symmetry is enhanced.
For the central charge two, we find the $\mathrm{SU}(2)^2\times \mathrm{SU}(2)^2$ enhanced symmetry manifests itself in the CFTs constructed by codes over $\BZ_{k^2}$ generated by \eqref{eq:n2code}. Naturally, a question exists whether we can find a code that corresponds to an $\mathrm{SU}(3)\times \mathrm{SU}(3)$ enhanced symmetry.
Recently, affine symmetry of Narain CFTs was discussed to eliminate fake partition functions of code CFTs~\cite{Dymarsky:2022kwb}.
Their approach would be helpful to read off the extended symmetry of Narain code CFTs.
We note that there is a manifest $\SU(2)^n_1 = \U(1)_2^n$ affine symmetry in the Narain code CFTs of~\cite{Henriksson:2022dml}, where they use the identification between~$(p_L,p_R)$ and $\lambda$ as well as the lattice metric that are different from~\cite{Dymarsky:2020qom} and this paper.
It would be interesting to extend their construction from binary to more general codes.

\acknowledgments
The work of T.\,N. was supported in part by the JSPS Grant-in-Aid for Scientific Research (C) No.19K03863, Grant-in-Aid for Scientific Research (A) No.\,21H04469, and
Grant-in-Aid for Transformative Research Areas (A) ``Extreme Universe''
No.\,21H05182 and No.\,21H05190.
The research of T.\,O. was supported in part by Grant-in-Aid for Transformative Research Areas (A) ``Extreme Universe'' No.\,21H05190.
The work of K.\,K. and S.\,Y. was supported by FoPM, WINGS Program, the University of Tokyo.
The work of K.\,K. was supported by JSPS KAKENHI Grant-in-Aid for JSPS fellows Grant No. 23KJ0436.

\appendix

\section{Proof of Proposition \ref{prop:even_sf_finite}}
\label{app:proof}
In this appendix, we give a proof of Proposition \ref{prop:even_sf_finite}.

\begin{proof}
\textbf{(Self-duality)} First, we want to prove that the dual of the lattice constructed from the code is equal to the lattice constructed from the dual code, i.e., $\Lambda(C)^\ast=\Lambda(C^\perp)$. From the definition of dual lattice and $\Lambda(C)$, for $\lambda'=(u_{1,1},\dots,u_{N,m})/\sqrt{p}\in\BR^{Nm}$,
\begin{subequations}
\begin{align}
    \lambda'\in\Lambda(C)^\ast \Leftrightarrow\;& \forall \lambda\in\Lambda(C),\ b(\lambda,\lambda')\in\BZ \nonumber \\
    \Leftrightarrow\;& \forall c\in C,\, l\in\BZ^{Nm},\ \sum_{i,j=1}^N \sum_{t,s=1}^{m} \left(\frac{1}{p}\,\iota(c_{i,t})+l_{i,t}\right) u_{j,s}\, w_{i,t,j,s} \in\BZ \nonumber \\
    \Leftrightarrow\;& \forall c\in C,\ \sum_{i,t,j,s} \iota(c_{i,t})\, u_{j,s}\, w_{i,t,j,s} \in p\,\BZ \quad\text{and} \label{eq:dual_condition_c} \\
    &\forall l\in\BZ^{Nm},\ \sum_{i,t,j,s} l_{i,t}\, u_{j,s} \,w_{i,t,j,s} \in\BZ \,. \label{eq:dual_condition_l}
\end{align}
\end{subequations}

From the unimodularity of $g$, \eqref{eq:dual_condition_l} is equivalent to $\forall j,s,\ u_{j,s} \in\BZ$ and then
\begin{equation}
    \exists c'\in \BF_{p^m}^N,\, l'\in\BZ^{Nm},\qquad u_{j,s}=\iota(c'_{j,s})+p\,l'_{j,s} \,.
\end{equation}
In this notation, from the condition \eqref{eq:metrics_condition},
\begin{align} \label{eq:products_real_Fpm}
\begin{aligned}
    &\sum_{i,t,j,s}\iota(c_{i,t})\,\iota(c'_{j,s})\,w_{i,t,j,s} \in p\,\BZ \\
    \,\Leftrightarrow\;& \sum_{i,t,j,s}\iota(c_{i,t})\,\iota(c'_{j,s})\,\iota\Bigl(\beta\bigl((0^{i-1},\, e_t,\, 0^{N-i}),(0^{j-1},\, e_s,\, 0^{N-j})\bigr)\Bigr) \in p\,\BZ \\
    \Leftrightarrow\;& \sum_{i,t,j,s} \beta\bigl((0^{i-1},\, c_{i,t}e_t,\, 0^{N-i}),(0^{j-1},\, c'_{j,s}e_s,\, 0^{N-j})\bigr) = \beta(c,c') = 0
\end{aligned}
\end{align}
and thus \eqref{eq:dual_condition_c} becomes $c'\in C^\perp$. Combining these results, we get
\begin{align}
\begin{aligned}
    \lambda'\in\Lambda(C)^\ast
    \Leftrightarrow\;& \exists c'\in C^\perp,\, l'\in\BZ^{Nm},\ \lambda'=\frac{1}{\sqrt{p}} \left(\iota(c'_{1,1})+p\,l'_{1,1},\dots,\iota(c'_{n,m})+p\,l'_{N,m}\right) \\
    \Leftrightarrow\;& \lambda'\in\Lambda(C^\perp) \,,
\end{aligned}
\end{align}
which means that $\Lambda(C)^\ast=\Lambda(C^\perp)$. Therefore, it follows that
\begin{equation}
    \Lambda(C)=\Lambda(C)^\ast \,\Leftrightarrow\, \Lambda(C)=\Lambda(C^\perp) \,\Leftrightarrow\, C=C^\perp \,.
\end{equation}

\textbf{(Evenness)}
From the definition of even lattice,
\begin{align}
\begin{aligned}
    \Lambda(C) \text{ is even}
    \Leftrightarrow\;& \forall \lambda\in\Lambda(C),\quad b(\lambda,\lambda)\in2\BZ \\
    \Leftrightarrow\;& \forall c\in C,\, l\in\BZ^{Nm},\quad \frac{1}{p}\sum_{i,t,j,s}(\iota(c_{i,t})+p\,l_{i,t})(\iota(c_{j,s})+p\,l_{j,s})\,w_{i,t,j,s}\in2\BZ \,.
\end{aligned}
\end{align}
Since $g$ is symmetric and unimodular, $w_{i,t,j,s}=w_{j,s,i,t}\in\BZ$ and thus the sum of the terms from $\iota(c_{i,t})\,l_{j,s}$ and $l_{i,t}\,\iota(c_{j,s})$ is always even. Therefore, the evenness can be divided into two parts as
\begin{subequations}
\begin{align}
    &\forall c\in C,\qquad \frac{1}{p}\sum_{i,t,j,s} \iota(c_{i,t})\,\iota(c_{j,s})\,w_{i,t,j,s}\in2\BZ \quad\text{and} \label{eq:even_condition_c} \\
    &\forall l\in\BZ^{Nm},\qquad p \sum_{i,t,j,s} l_{i,t}\,l_{j,s}\,w_{i,t,j,s} \in 2\mathbb{Z} \,. \label{eq:even_condition_l}
\end{align}
\end{subequations}
For $p=2$, \eqref{eq:even_condition_l} is automatically satisfied and \eqref{eq:even_condition_c} becomes \eqref{eq:evenness_p2} by transposing $1/p$. Note that this cannot be expressed in the language of $\BF_2$ since $0$ and $2 \pmod 4$ are indistinguishable on $\BF_2$. For odd prime $p$, since
\begin{equation} \label{eq:llw}
    \sum_{i,t,j,s} l_{i,t}\,l_{j,s}\,w_{i,t,j,s} = \sum_{(i,t)\neq(j,s)} l_{i,t}\,l_{j,s}\,w_{i,t,j,s} + \sum_{i,t} (l_{i,t})^2 \,w_{i,t,i,t}
\end{equation}
and its first term is even from $w_{i,t,j,s}=w_{j,s,i,t}\in\BZ$, \eqref{eq:even_condition_l} is equivalent to $\forall i,t,\ w_{i,t,i,t}\in2\BZ$. In this case, by \eqref{eq:products_real_Fpm} and the same discussion as \eqref{eq:llw}, \eqref{eq:even_condition_c} becomes \eqref{eq:evenness_podd}.
\end{proof}

\bibliographystyle{JHEP}
\bibliography{QEC_CFT_prime_power}

\providecommand{\href}[2]{#2}\begingroup\raggedright\begin{thebibliography}{10}

\bibitem{Narain:1985jj}
K.~S. Narain, {\it {New Heterotic String Theories in Uncompactified Dimensions
  \ensuremath{<} 10}},  {\em Phys. Lett. B} {\bf 169} (1986) 41--46.

\bibitem{Narain:1986am}
K.~S. Narain, M.~H. Sarmadi, and E.~Witten, {\it {A Note on Toroidal
  Compactification of Heterotic String Theory}},  {\em Nucl. Phys. B} {\bf 279}
  (1987) 369--379.

\bibitem{Polchinski:1998rq}
J.~Polchinski, {\em {String theory. Vol. 1: An introduction to the bosonic
  string}}.
\newblock Cambridge University Press, 2007.

\bibitem{Dymarsky:2020qom}
A.~Dymarsky and A.~Shapere, {\it {Quantum stabilizer codes, lattices, and
  CFTs}},  {\em JHEP} {\bf 21} (2020) 160,
  [\href{http://arxiv.org/abs/2009.01244}{{\tt arXiv:2009.01244}}].

\bibitem{Gottesman:1996rt}
D.~Gottesman, {\it {A Class of quantum error correcting codes saturating the
  quantum Hamming bound}},  {\em Phys. Rev. A} {\bf 54} (1996) 1862,
  [\href{http://arxiv.org/abs/quant-ph/9604038}{{\tt quant-ph/9604038}}].

\bibitem{conway2013sphere}
J.~H. Conway and N.~J.~A. Sloane, {\em Sphere packings, lattices and groups},
  vol.~290.
\newblock Springer Science \& Business Media, 2013.

\bibitem{Kawabata:2022jxt}
K.~Kawabata, T.~Nishioka, and T.~Okuda, {\it {Narain CFTs from qudit stabilizer
  codes}},  {\em SciPost Physics Core} {\bf 6} (2023), no.~2 035,
  [\href{http://arxiv.org/abs/2212.07089}{{\tt arXiv:2212.07089}}].

\bibitem{Calderbank:1995dw}
A.~R. Calderbank and P.~W. Shor, {\it {Good quantum error correcting codes
  exist}},  {\em Phys. Rev. A} {\bf 54} (1996) 1098,
  [\href{http://arxiv.org/abs/quant-ph/9512032}{{\tt quant-ph/9512032}}].

\bibitem{Steane:1996va}
A.~M. Steane, {\it {Simple quantum error correcting codes}},  {\em Phys. Rev.
  A} {\bf 54} (1996) 4741, [\href{http://arxiv.org/abs/quant-ph/9605021}{{\tt
  quant-ph/9605021}}].

\bibitem{Dymarsky:2020bps}
A.~Dymarsky and A.~Shapere, {\it {Solutions of modular bootstrap constraints
  from quantum codes}},  {\em Phys. Rev. Lett.} {\bf 126} (2021), no.~16
  161602, [\href{http://arxiv.org/abs/2009.01236}{{\tt arXiv:2009.01236}}].

\bibitem{Henriksson:2022dnu}
J.~Henriksson, A.~Kakkar, and B.~McPeak, {\it {Narain CFTs and quantum codes at
  higher genus}},  {\em JHEP} {\bf 04} (2023) 011,
  [\href{http://arxiv.org/abs/2205.00025}{{\tt arXiv:2205.00025}}].

\bibitem{Dymarsky:2022kwb}
A.~Dymarsky and R.~R. Kalloor, {\it {Fake Z}},  {\em JHEP} {\bf 06} (2023) 043,
  [\href{http://arxiv.org/abs/2211.15699}{{\tt arXiv:2211.15699}}].

\bibitem{Furuta:2022ykh}
Y.~Furuta, {\it {Relation between spectra of Narain CFTs and properties of
  associated boolean functions}},  {\em JHEP} {\bf 09} (2022) 146,
  [\href{http://arxiv.org/abs/2203.11643}{{\tt arXiv:2203.11643}}].

\bibitem{Angelinos:2022umf}
N.~Angelinos, D.~Chakraborty, and A.~Dymarsky, {\it {Optimal Narain CFTs from
  codes}},  {\em JHEP} {\bf 11} (2022) 118,
  [\href{http://arxiv.org/abs/2206.14825}{{\tt arXiv:2206.14825}}].

\bibitem{Dymarsky:2020pzc}
A.~Dymarsky and A.~Shapere, {\it {Comments on the holographic description of
  Narain theories}},  {\em JHEP} {\bf 10} (2021) 197,
  [\href{http://arxiv.org/abs/2012.15830}{{\tt arXiv:2012.15830}}].

\bibitem{Maloney:2020nni}
A.~Maloney and E.~Witten, {\it {Averaging over Narain moduli space}},  {\em
  JHEP} {\bf 10} (2020) 187, [\href{http://arxiv.org/abs/2006.04855}{{\tt
  arXiv:2006.04855}}].

\bibitem{Afkhami-Jeddi:2020ezh}
N.~Afkhami-Jeddi, H.~Cohn, T.~Hartman, and A.~Tajdini, {\it {Free partition
  functions and an averaged holographic duality}},  {\em JHEP} {\bf 01} (2021)
  130, [\href{http://arxiv.org/abs/2006.04839}{{\tt arXiv:2006.04839}}].

\bibitem{aharony2023holographic}
O.~Aharony, A.~Dymarsky, and A.~D. Shapere, {\it Holographic description of
  narain cfts and their code-based ensembles},
  \href{http://arxiv.org/abs/2310.06012}{{\tt arXiv:2310.06012}}.

\bibitem{Gaiotto:2018ypj}
D.~Gaiotto and T.~Johnson-Freyd, {\it {Holomorphic SCFTs with small index}},
  {\em Can. J. Math.} {\bf 74} (2022), no.~2 573--601,
  [\href{http://arxiv.org/abs/1811.00589}{{\tt arXiv:1811.00589}}].

\bibitem{Buican:2021uyp}
M.~Buican, A.~Dymarsky, and R.~Radhakrishnan, {\it {Quantum codes, CFTs, and
  defects}},  {\em JHEP} {\bf 03} (2023) 017,
  [\href{http://arxiv.org/abs/2112.12162}{{\tt arXiv:2112.12162}}].

\bibitem{Henriksson:2022dml}
J.~Henriksson and B.~McPeak, {\it {Averaging over codes and an $SU(2)$ modular
  bootstrap}},  \href{http://arxiv.org/abs/2208.14457}{{\tt arXiv:2208.14457}}.

\bibitem{Yahagi:2022idq}
S.~Yahagi, {\it {Narain CFTs and error-correcting codes on finite fields}},
  {\em JHEP} {\bf 08} (2022) 058, [\href{http://arxiv.org/abs/2203.10848}{{\tt
  arXiv:2203.10848}}].

\bibitem{Kawabata:2023nlt}
K.~Kawabata and S.~Yahagi, {\it {Fermionic CFTs from classical codes over
  finite fields}},  {\em JHEP} {\bf 05} (2023) 096,
  [\href{http://arxiv.org/abs/2303.11613}{{\tt arXiv:2303.11613}}].

\bibitem{ashikhmin2001nonbinary}
A.~Ashikhmin and E.~Knill, {\it Nonbinary quantum stabilizer codes},  {\em IEEE
  Transactions on Information Theory} {\bf 47} (2001), no.~7 3065--3072,
  [\href{http://arxiv.org/abs/quant-ph/0005008}{{\tt quant-ph/0005008}}].

\bibitem{nadella2012stabilizer}
S.~Nadella and A.~Klappenecker, {\it Stabilizer codes over frobenius rings},
  in {\em 2012 IEEE International Symposium on Information Theory Proceedings},
  pp.~165--169, IEEE, 2012.

\bibitem{guenda2014quantum}
K.~Guenda and T.~A. Gulliver, {\it Quantum codes over rings},  {\em
  International Journal of Quantum Information} {\bf 12} (2014), no.~04
  1450020.

\bibitem{Ginsparg:1987eb}
P.~H. Ginsparg, {\it {Curiosities at c = 1}},  {\em Nucl. Phys. B} {\bf 295}
  (1988) 153--170.

\bibitem{DiFrancesco:1997nk}
P.~Di~Francesco, P.~Mathieu, and D.~Senechal, {\em {Conformal Field Theory}}.
\newblock Graduate Texts in Contemporary Physics. Springer-Verlag, New York,
  1997.

\bibitem{Gukov:2002nw}
S.~Gukov and C.~Vafa, {\it {Rational conformal field theories and complex
  multiplication}},  {\em Commun. Math. Phys.} {\bf 246} (2004) 181--210,
  [\href{http://arxiv.org/abs/hep-th/0203213}{{\tt hep-th/0203213}}].

\bibitem{Dymarsky:2021xfc}
A.~Dymarsky and A.~Sharon, {\it {Non-rational Narain CFTs from codes over
  F$_{4}$}},  {\em JHEP} {\bf 11} (2021) 016,
  [\href{http://arxiv.org/abs/2107.02816}{{\tt arXiv:2107.02816}}].

\bibitem{Furuta:2023xwl}
Y.~Furuta, {\it {On the Rationality and the Code Structure of a Narain CFT, and
  the Simple Current Orbifold}},  \href{http://arxiv.org/abs/2307.04190}{{\tt
  arXiv:2307.04190}}.

\bibitem{PhysRevA.52.R2493}
P.~W. Shor, {\it Scheme for reducing decoherence in quantum computer memory},
  {\em Phys. Rev. A} {\bf 52} (Oct, 1995) R2493--R2496.

\bibitem{steane1996multiple}
A.~Steane, {\it Multiple-particle interference and quantum error correction},
  {\em Proceedings of the Royal Society of London. Series A: Mathematical,
  Physical and Engineering Sciences} {\bf 452} (1996), no.~1954 2551--2577.

\bibitem{knill1996group}
E.~Knill, {\it Group representations, error bases and quantum codes},
  \href{http://arxiv.org/abs/quant-ph/9608049}{{\tt quant-ph/9608049}}.

\bibitem{knill1996non}
E.~Knill, {\it Non-binary unitary error bases and quantum codes},
  \href{http://arxiv.org/abs/quant-ph/9608048}{{\tt quant-ph/9608048}}.

\bibitem{rains1999nonbinary}
E.~M. Rains, {\it Nonbinary quantum codes},  {\em IEEE Transactions on
  Information Theory} {\bf 45} (1999), no.~6 1827--1832,
  [\href{http://arxiv.org/abs/quant-ph/9703048}{{\tt quant-ph/9703048}}].

\bibitem{klappenecker2012nice}
A.~Klappenecker, {\it Nice nearrings},  in {\em 2012 IEEE International
  Symposium on Information Theory Proceedings}, pp.~170--173, IEEE, 2012.

\bibitem{Wood2011APPLICATIONSOF}
J.~A. Wood, {\it Applications of finite frobenius rings to the foundations of
  algebraic coding theory},  {\em Proceedings of the 44th Symposium on Ring
  Theory and Representation Theory} (2011).

\bibitem{dougherty2017algebraic}
S.~T. Dougherty, {\em Algebraic coding theory over finite commutative rings}.
\newblock Springer, 2017.

\bibitem{hirano1997admissible}
Y.~Hirano, {\it On admissible rings},  {\em Indagationes Mathematicae} {\bf 8}
  (1997), no.~1 55--59.

\bibitem{wood1999duality}
J.~A. Wood, {\it Duality for modules over finite rings and applications to
  coding theory},  {\em American journal of Mathematics} (1999) 555--575.

\bibitem{GHEORGHIU2014505}
V.~Gheorghiu, {\it Standard form of qudit stabilizer groups},  {\em Physics
  Letters A} {\bf 378} (2014), no.~5 505--509.

\bibitem{Calderbank:1996aj}
A.~R. Calderbank, E.~M. Rains, P.~Shor, and N.~J. Sloane, {\it Quantum error
  correction via codes over gf (4)},  {\em IEEE Transactions on Information
  Theory} {\bf 44} (1998), no.~4 1369--1387,
  [\href{http://arxiv.org/abs/quant-ph/9608006}{{\tt quant-ph/9608006}}].

\bibitem{Calderbank:1996hm}
A.~R. Calderbank, E.~M. Rains, N.~J.~A. Sloane, and P.~W. Shor, {\it {Quantum
  error correction and orthogonal geometry}},  {\em Phys. Rev. Lett.} {\bf 78}
  (1997) 405--408, [\href{http://arxiv.org/abs/quant-ph/9605005}{{\tt
  quant-ph/9605005}}].

\bibitem{calderbank1996good}
A.~R. Calderbank and P.~W. Shor, {\it Good quantum error-correcting codes
  exist},  {\em Phys. Rev. A} {\bf 54} (1996), no.~2 1098,
  [\href{http://arxiv.org/abs/quant-ph/9512032}{{\tt quant-ph/9512032}}].

\bibitem{Friedan:1983xq}
D.~Friedan, Z.-a. Qiu, and S.~H. Shenker, {\it {Conformal Invariance, Unitarity
  and Two-Dimensional Critical Exponents}},  {\em Phys. Rev. Lett.} {\bf 52}
  (1984) 1575--1578.

\bibitem{Moore:1988uz}
G.~W. Moore and N.~Seiberg, {\it {Polynomial Equations for Rational Conformal
  Field Theories}},  {\em Phys. Lett. B} {\bf 212} (1988) 451--460.

\bibitem{Moore:1988qv}
G.~W. Moore and N.~Seiberg, {\it {Classical and Quantum Conformal Field
  Theory}},  {\em Commun. Math. Phys.} {\bf 123} (1989) 177.

\bibitem{Moore:1989yh}
G.~W. Moore and N.~Seiberg, {\it {Taming the Conformal Zoo}},  {\em Phys. Lett.
  B} {\bf 220} (1989) 422--430.

\bibitem{leech1971sphere}
J.~Leech and N.~Sloane, {\it Sphere packings and error-correcting codes},  {\em
  Canadian Journal of Mathematics} {\bf 23} (1971), no.~4 718--745.

\bibitem{Kawabata:2023iss}
K.~Kawabata, T.~Nishioka, and T.~Okuda, {\it {Narain CFTs from quantum codes
  and their $\mathbb{Z}_2$ gauging}},
  \href{http://arxiv.org/abs/2308.01579}{{\tt arXiv:2308.01579}}.

\bibitem{Shor-zoo}
V.~V. Albert, P.~Faist, and contributors, ``$[[9,1,3]]$ shor code.''
\newblock \url{https://errorcorrectionzoo.org/c/shor_nine}.

\bibitem{Steane-zoo}
V.~V. Albert, P.~Faist, and contributors, ``$[[7,1,3]]$ shor code.''
\newblock \url{https://errorcorrectionzoo.org/c/steane}.

\bibitem{Dixon:1988qd}
L.~J. Dixon, P.~H. Ginsparg, and J.~A. Harvey, {\it {Beauty and the Beast:
  Superconformal Symmetry in a Monster Module}},  {\em Commun. Math. Phys.}
  {\bf 119} (1988) 221--241.

\bibitem{Kawabata:2023rlt}
K.~Kawabata and S.~Yahagi, {\it {Elliptic genera from classical
  error-correcting codes}},  \href{http://arxiv.org/abs/2308.12592}{{\tt
  arXiv:2308.12592}}.

\bibitem{Kawabata:2023usr}
K.~Kawabata, T.~Nishioka, and T.~Okuda, {\it {Supersymmetric conformal field
  theories from quantum stabilizer codes}},
  \href{http://arxiv.org/abs/2307.14602}{{\tt arXiv:2307.14602}}.

\bibitem{Moore:2023zmv}
G.~W. Moore and R.~K. Singh, {\it {Beauty And The Beast Part 2: Apprehending
  The Missing Supercurrent}},  \href{http://arxiv.org/abs/2309.02382}{{\tt
  arXiv:2309.02382}}.

\bibitem{Barbar:2023ncl}
A.~Barbar, A.~Dymarsky, and A.~D. Shapere, {\it {Global Symmetries, Code
  Ensembles, and Sums Over Geometries}},
  \href{http://arxiv.org/abs/2310.13044}{{\tt arXiv:2310.13044}}.

\bibitem{harada2009classification}
M.~Harada and A.~Munemasa, {\it On the classification of self-dual-codes},  in
  {\em Cryptography and Coding: 12th IMA International Conference, Cryptography
  and Coding 2009, Cirencester, UK, December 15-17, 2009. Proceedings 12},
  pp.~78--90, Springer, 2009.

\bibitem{harada2016classification}
M.~Harada and A.~Munemasa, {\it On the classification of self-dual
  $\mathbb{Z}_k$-codes ii},  {\em Interdisciplinary Information Sciences} {\bf
  22} (2016), no.~1 81--85.

\bibitem{munemasa_database}
M.~Harada and A.~Munemasa, ``Database of self-dual codes.''
  \url{https://www.math.is.tohoku.ac.jp/~munemasa/selfdualcodes.htm}.

\bibitem{balmaceda2008mass}
J.~M.~P. Balmaceda, R.~A.~L. Betty, and F.~R. Nemenzo, {\it Mass formula for
  self-dual codes over $\mathbb{Z}_{p^2}$},  {\em Discrete Mathematics} {\bf
  308} (2008), no.~14 2984--3002.

\bibitem{betty2008mass}
R.~A. Betty and A.~Munemasa, {\it Mass formula for self-orthogonal codes over
  $\mathbb{Z}_{p^2}$},  {\em arXiv preprint arXiv:0805.2205} (2008).

\end{thebibliography}\endgroup
\end{document}